\documentclass[11pt]{amsart}

\usepackage[english]{babel} 
\usepackage{amsmath,amssymb,amsfonts,amsthm}
\usepackage{mathrsfs, esint, dsfont}
\usepackage{moreverb,rotating,graphics,float}
\usepackage{caption, subcaption}
\usepackage[colorlinks, linkcolor={blue}, citecolor={red}]{hyperref}
\usepackage{framed,fancybox}
\usepackage{enumerate}
\usepackage{csquotes}
\usepackage{slashed}

\usepackage{a4wide} 

\usepackage[backend=biber,
style=alphabetic,
maxbibnames=99,
doi=false, url=false, giveninits=true, date=year]{biblatex}

\AtEveryBibitem{
    \clearfield{issn}
}    
\numberwithin{equation}{section} 
\newtheorem{theorem}{Theorem}[section]
\newtheorem{proposition}[theorem]{Proposition}
\newtheorem{remark}[theorem]{Remark}
\newtheorem{lemma}[theorem]{Lemma}

\newtheorem{definition}[theorem]{Definition}

\newtheorem{example}[theorem]{Example}

\newcommand\1{{\mathds 1}}

\def\C{{\mathbb C}}
\def\CC{{\mathbb C}}

\def\bbI{{\mathbb I}}

\def\R{{\mathbb R}}
\def\RR{{\mathbb R}}

\def\Z{{\mathbb Z}}

\def\rd{{\mathrm{d}}}
\def\re{{\mathrm{e}}}
\def\ri{{\mathrm{i}}}

\def\cA{{\mathcal A}}

\def\cF{{\mathcal F}}
\def\cG{{\mathcal G}}
\def\cH{{\mathcal H}}

\def\cK{{\mathcal K}}
\def\cL{{\mathcal L}}
\def\cM{{\mathcal M}}

\def\cS{{\mathcal S}}

\def\cU{{\mathcal U}}

\def\rA{{\rm A}}
\def\rAIII{{\rm AIII}}
\def\rAI{{\rm AI}}
\def\rBDI{{\rm BDI}}
\def\rD{{\rm D}}
\def\rDIII{{\rm DIII}}
\def\rAII{{\rm AII}}
\def\rCII{{\rm CII}}

\def\U{{\rm U}}
\def\O{{\rm O}}

\def\Sp{{\rm Sp}}
\def\Index{{\rm Index}}

\def\bra{\langle}
\def\ket{\rangle}
\def\Tr{{\rm Tr}}
\def\Ker{{\rm Ker }}
\def\Ran{{\rm Ran }}

\def\dim{{\rm dim }}
\def\det{{\rm det}}
\def\Pf{{\rm Pf}}

\def\Dirac{\slashed{D}} 

\newcommand{\myfootnote}[1]{
    \renewcommand{\thefootnote}{}
    \footnotetext{\scriptsize#1}
    \renewcommand{\thefootnote}{\arabic{footnote}}
}

\title{Topological junctions for one--dimensional systems}
\author{David Gontier \qquad Clément Tauber}
\date{\today}

\begin{document}

\myfootnote{David Gontier: CEREMADE, Université Paris-Dauphine, PSL University,75016 Paris, France \& ENS/PSL University, DMA, F-75005, Paris, France;\\
    email: \href{gontier@ceremade.dauphine.fr}{gontier@ceremade.dauphine.fr}}
\myfootnote{Clément Tauber: CEREMADE, Université Paris-Dauphine, PSL University,75016 Paris, France;\\
    email: \href{tauber@ceremade.dauphine.fr}{tauber@ceremade.dauphine.fr}}
%

\begin{abstract}    
    We study and classify the emergence of protected edge modes at the junction of one-dimensional materials. Using symmetries of Lagrangian planes in boundary symplectic spaces, we present a novel proof of the periodic table of topological insulators in one dimension. We show that edge modes necessarily arise at the junction of two materials having different topological indices. Our approach provides a systematic framework for understanding symmetry--protected modes in one--dimension. It does not rely on periodic nor ergodicity and covers a wide range of operators which includes both continuous and discrete models. 
    
    \bigskip
    \noindent \sl \copyright~2024 by the authors. This paper may be reproduced, in its entirety, for non-commercial purposes.
\end{abstract}


\maketitle    



\section{Introduction}

The classification of bulk gapped Hamiltonians based on their symmetries results in the renowned {\em tenfold way of topological insulators}. This framework identifies 10 fundamental symmetry classes that encompass three key physical symmetries: time-reversal (T), charge conjugation (C), and chiral (S) symmetry. Each symmetry class is associated with a unique topological index, which enjoys the following important property: for two insulators within the same symmetry class but with different indices, any symmetry-preserving transformation between them must involve a metallic phase where the energy gap closes. The classification result was first tackled by Altland and Zirnbauer in~\cite{AltZir-97} in the context of random matrices (see also~\cite{HeiHucZir-05}), where Cartan's classification of symmetric spaces from~\cite{Car-26, Car-27} naturally appears and provides a label for each symmetry class. The classification was then developed simultaneously by Kitaev in~\cite{Kit-09} and Ryu, Schnyder, Furusaki and Ludwig in~\cite{RyuSchFur-10}, where the full table of topological insulators appears. It was then extended to various systems using different approaches from $K$-theory to functional analysis and is still an active topic of research \cite{FreMor-13,ProSch-16,Thi-16,BouCarAla-16,GroSch-16,KatKom-18,AllMaxZir-20,AvrTur-22,GonMonPer-22,ChuSha-23}.

\medskip

Topological insulators have the property that protected edge modes appear when gapped bulk Hamiltonians are restricted to half-infinite samples. This is known as {\em bulk-edge correspondence}~\cite{Hal-82, Hat-93, Hat-93a}. An intuitive explanation is the following: cutting such a system induces a natural transition between the material and the surrounding vacuum. This transition can be visualized as a large box gradually moving towards the empty space. If this cut is implemented as a smooth transition at very large scales, the system within the box consistently maintains the characteristics of a bulk material. Since the void has null index, and in the case where the topological insulator has a non null index, the gap must close somewhere. The corresponding eigenstates come from the cut, hence the denomination {\em edge modes}. 

\medskip

Edge modes are usually studied directly in half-infinite samples but, to the best of our knowledge,  the previous argument has not been made rigorous except for some specific models or in a given symmetry class \cite{GraPor-13, Dro-21,Bal-22,BalBecDro-23,RosTar-24, Bal-}. More generally, one can ask the following question: {\em consider a junction between two topological insulators in the same symmetry class. If their two indices differ, do we always observe edge modes at the junction?} Note that this question differs from the original classification, where deformations are considered solely among gapped/bulk systems. When considering junctions, we must leave the realms of gapped operators. The goal of this paper is to give a positive and systematic answer to this question in dimension $d = 1$. We prove that protected edge modes must appear at the junction of two half-chains having different indices. In the process, we give a straightforward procedure to classify a wide range of one-dimensional operators based on their symmetries.

\medskip

Our classification in $d=1$ relies on efficient concepts from ordinary differential equations, originally developed in \cite{Gon-20, Gon-23}. To each bulk Hamiltonian with a gap, one can associate a Lagrangian plane in a symplectic boundary space, and a unitary matrix, according to Leray's theorem \cite{Ler-78} (see Theorem~\ref{th:Leray} below). Classifying symmetric bulk Hamiltonians reduces to classifying symmetric Lagrangian planes, which further amounts to classifying subsets of symmetric unitary matrices. This identification recovers Cartan's symmetric spaces in a natural way. We insist on the fact that we do not assume any periodicity or ergodicity of the Hamiltonians. In particular, disordered and inhomogeneous media are naturally covered within the approach. Our result would extend to higher dimension if we moreover assume periodicity in the longitudinal directions of the junction. In this case, after applying an appropriate Bloch transform, we recover a family of models on a line (or on a ``channel'' for continuous models, see~\cite{Gon-23}). 

\bigskip

The paper is organized as follows. In Section~\ref{sec:main} we explain the relationship between one dimensional differential operators, Lagrangian planes, and unitary matrices. We then state the two main results of the manuscript. Our first main result, Theorem~\ref{thm:classification} concerns the classification of Lagrangian planes respecting symmetries. Our second main result, Theorem~\ref{thm:BEC} shows the appearance of edge modes for symmetry protected topological junctions. We also present various classes of models for which our theory applies. We then prove the first result in Section~\ref{sec:classification}, and the second one in Section~\ref{sec:proof:junctions}. Finally, Section~\ref{sec:application_cst_Dirac} is a case study where we apply our results to junctions between two continuous Dirac operators with constant potentials. This simple case covers all the non-trivial symmetry classes.

\subsection*{Notation}
Let us gather here the conventions that we use for matrices. We denote by $\cM_N(\C)$ the set of $N \times N$ matrices with complex entries. For $z \in \C$, we denote by $z^* = \overline{z} \in \C$ the complex conjugate. For $A \in \cM_N(\C)$, we denote by $A^*$ its adjoint $(A^*)_{ij} = \overline{A_{ji}}$, by $A^T$ its transpose $(A^T)_{ij} = A_{ji}$, and by $\overline{A}$ its conjugate $(\overline{A})_{ij} = \overline{A_{ij}}$. In particular, $A^* = \overline{A}^T$. 

\medskip

We will encounter the following sets of matrices. First, the unitary and orthogonal groups
\[
    \U(N)  := \left\{ M \in \cM_N(\C), \ M^{-1} = M^* \right\}, \qquad    
    \O(N) := \left\{ M \in \U_N(\C), \ \overline{M} = M \right\}.
\]
Then the complex and real symmetric and anti--symmetric matrices:
\begin{align*}
     \cS_N(\C)  := \{M \in \cM_N(\C), \  M^*= M \}, & \qquad
    \cS_N^\R(\C) := \{M \in \cM_N(\C), \ M^T= M\}, \\
     \cA_N(\C)  := \{M \in \cM_N(\C), \ M^*= - M\}, & \qquad  
   	\cA_N^\R(\C) := \{M \in \cM_N(\C), \  M^T= - M\}.
\end{align*}
Finally, in the case where $N = 2n$ is even, the symplectic matrices
\begin{equation} \label{eq:def:symplectic_Omega}
    \Sp(2n,\C) := \{M \in\cM_N(\C), \  M^T \Omega M = \Omega\},  \qquad
    \Omega := \begin{pmatrix}
    0 & \bbI_n \\ - \bbI_n & 0
\end{pmatrix}.
\end{equation}
and the symplectic group $\Sp(n) := \Sp(2n,\C) \cap \U(2n)$.

\section{General setting and main result\label{sec:main}}

\subsection{From operators to Lagrangian planes and unitaries}
Let us first explain the relationships between differential operators, Lagrangian planes in a boundary symplectic space, and unitary matrices. In what follows, we restrict our attention to one-dimensional systems. 

\subsubsection{From operator to Lagrangian planes}

Our method relies on the following observations (see also~\cite{Gon-20, Gon-23}). We consider a general differential operator $\cL$ of order $p$ on $\C^M$ which generates a densely defined self-adjoint operator $H$ on $L^2(\R, \C^M)$, describing the dynamics of independent electrons in one-dimension. We emphasize that, compared to previous works, we do not assume any particular structure for the differential operator $\cL$, apart that
\begin{enumerate}
    \item it defines a self-adjoint operator $H$,
    \item it defines a symplectic form with~\eqref{eq:def:omega} (see below),
    \item the spectrum of $H$ is not the full line $\R$ (so we can find $E \in \R \setminus \sigma(H)$).
\end{enumerate}
Whenever these conditions are satisfied, we will prove that we can associate an index to the operator $H$ at energy $E$. Of course, checking these points, and specifically point (3), might be difficult if no further structure are assumed. In the context of Schrödinger operators for instance, points (1) and (2) are automatically satisfied for $\cL = - \partial_{tt}^2 + V(t)$ whenever $V \in L^\infty (\R)$ is bounded, but checking whether there are gaps in the spectrum can be a delicate question.

\medskip

First, Cauchy's theory for ordinary differential equations states that for any $E \in \R$, the set of solutions $\cS_E=\{ \psi : \R \to \C^M \, | \, \cL \psi = E \psi\}$ forms a vector space of dimension $Mp$, parametrized {\em e.g.} by the initial conditions $(\psi(0), \psi'(0), \ldots, \psi^{(p-1)}(0)) \in (\C^M)^p$. This means that the evaluation map
\begin{equation} \label{eq:def:Tr}
    \Tr : \left\lbrace \begin{array}{ccc} \cS_E & \to  & (\C^M)^p\\  \psi &\mapsto &(\psi(0), \psi'(0), \ldots, \psi^{(p-1)}(0)) \end{array} \right.
\end{equation}
is a linear bijection. Moreover, a solution $\psi \in \cS_E$ is an eigenstate of $H$ iff $\psi$ is square integrable, that is
\[
    \Ker(H - E) = \cS_E \cap L^2(\R, \C^M).
\]
Such a function $\psi$ is square integrable on the real line $\R$ iff it is square integrable at $+ \infty$ and at $-\infty$. This motivates to introduce the following vector subspaces of $(\C^M)^p$:
\begin{equation} \label{eq:def:ell_E}
    \begin{cases}
        \ell_E^+ & := \Tr \left(  \cS_E \cap L^2(\R^+, \C^M) \right) \\ 
        \ell_E^- & := \Tr \left(  \cS_E \cap L^2(\R^-, \C^M) \right) \\ 
        \ell_E & :=  \Tr \left(  \cS_E \cap L^2(\R, \C^M) \right)
    \end{cases},
    \quad \text{so that} \quad \ell_E = \ell_E^+ \cap \ell_E^-.
\end{equation}
In other words, $\ell_E^+$ is the set of initial conditions at $t=0$ so that the corresponding Cauchy solution is square integrable at $+ \infty$. The fact that $\Tr$ is bijective implies that
\[
    \dim \,  \Ker \, (H - E) = \dim \, \ell_E = \dim \left( \ell_E^+ \cap \ell_E^- \right).
\]
The key point is that we can read the multiplicity of $E$ (a bulk quantity) solely from the crossing of two vector spaces defined at $t=0$ (which looks like an edge quantity). Thus $\ell_E^\pm$ naturally appear as a pivotal tool for bulk-edge correspondence. In addition, we note that the Cauchy problem $\cL \psi = E \psi$ on $\R^+$ is decoupled from the one on $\R^-$, so the vector spaces $\ell_E^\pm$ solely depends on the behaviour of $\cL$ on $\R^\pm$. These objects are therefore also adapted to study junctions between two materials, as we will see below.

\medskip

On the other hand, the self-adjointness and locality of $H$ induces a natural symplectic form on the boundary space $(\C^M)^p$. Recall that a symplectic form $\omega : \C^K \times \C^K \to \C$ is a non degenerate continuous sesquilinear form satisfying
\[
\forall x, y \in \C^K, \quad \omega(x,y) = - \overline{\omega(y,x)}.
\]
The usual Hilbertian structure of $\C^K$ shows that there is a linear map $J : \C^K \to \C^K$ so that
\begin{equation}\label{eq:defJ}
    \forall x, y \in \C^K, \qquad \omega(x, y) = \bra x, J y \ket_{\C^K}.
\end{equation}
The condition that $\omega$ is symplectic shows that $J$ is invertible and satisfies $J^* = -J$. 

\medskip

In our case, the symplectic form $\omega(\cdot, \cdot)$ is defined as follows. Let $x, y \in (\C^M)^p$, and let $\phi, \psi :  \R \to \C^M$ be {\em any} compactly supported smooth functions with $\Tr(\phi) = x$ and $\Tr(\psi) = y$. We set
\begin{equation} \label{eq:def:omega}
    \omega(x, y) := \int_0^{\infty} \Big( \bra \phi, \cL \psi \ket_{\C^M} - \bra \cL \phi,  \psi \ket_{\C^M} \Big)(t) \rd t
    = \int_0^{\infty} \Big( \bra \phi, H \psi \ket_{\C^M} - \bra H \phi,  \psi \ket_{\C^M} \Big)(t) \rd t.
\end{equation}
In most cases, the fact that $\omega(x, y)$ does not depend on the choices of $\psi$ and $\phi$, and the fact that $\omega$ is a symplectic form can be checked directly. Although we considered $\phi, \psi \in C^\infty_0(\R)$, the left--hand side of~\eqref{eq:def:omega} is a continuous bilinear form on $\left\{ \phi \in L^2(\R^+, \C^M), \ \cL \phi \in L^2(\R^+, \C^M)\right\}$, so formula~\eqref{eq:def:omega} also holds for $\phi$ and $\psi$ in this set (and in particular for $\phi, \psi \in \cS_E^+$). Here is a typical example (we give other examples below in Section~\ref{sec:examples}).

\begin{example}[Schrödinger operators] \label{ex:schrodinger}
Consider $\cL = - \partial_{tt}^2 + V(t)$ of order $p = 2$. We assume that $V(t)$ is pointwise a hermitian $M \times M$ matrix, uniformly bounded in $t \in \R$. Then the operator $H := - \partial_{tt}^2 + V(t)$ is self-adjoint with domain $H^2(\R, \C^M)$.  For all $\phi, \psi \in C^\infty_0(\R, \C^M)$, an integration by part gives
\begin{align*}
    \int_0^\infty \left[ \bra \phi, H \psi \ket_{\C^M} - \bra H \phi, \psi \ket_{\C^M} \right](t) \rd t & = -\int_0^\infty \left[ \bra \phi, \psi'' \ket_{\C^M} - \bra \phi'', \psi \ket_{\C^M} \right](t) \rd t  \\
    & = \bra \phi(0), \psi'(0) \ket_{\C^M} -  \bra \phi'(0), \psi(0) \ket_{\C^M} ,
\end{align*}
which is of the form $\omega(\Tr(\phi), \Tr(\psi))$, with the symplectic form $\omega : \C^{2M} \times \C^{2M} \to \C$ given by
\begin{equation}\label{eq:exmp_schrodinger_omega}
	\omega(x, y) = \bra x, J y \ket_{\C^{2M}}, \quad \text{with} \quad 
	J = \begin{pmatrix}
		0 & \bbI_M \\ - \bbI_M & 0    
	\end{pmatrix}.
\end{equation}
As claimed, $\omega(x, y)$ is independent of the choices of $\phi, \psi \in C^\infty_0(\R)$ as long as $x = \Tr(\phi)$ and $y = \Tr(\psi)$. So~\eqref{eq:def:omega} indeed defines a symplectic form for Schrödinger operators. Note that the symplectic form $\omega(\cdot, \cdot)$ is independent of the potential $V$. 
\end{example}

\begin{definition}
	A vector subspace $\ell \subset \C^K$ is a {\em Lagrangian plane} (for $\omega$) if it satisfies $\ell = \ell^\circ$, and is {\em isotropic} if $\ell \subset \ell^\circ$, where
	\begin{equation} \label{eq:def:Lagrangian_planes}
		\ell^\circ := \left\{ x \in \C^K, \quad \omega(x,y) = 0 \quad \text{for all} \quad y \in \ell \right\}.
	\end{equation}
\end{definition}

We have $\dim(\ell^\circ) = K - \dim(\ell)$. In particular, Lagrangian planes can exist only if $K$ is even, and in this case, they are of dimension $K/2$. For all $E \in \R$, the vector spaces $\ell_E^\pm$ are isotropic. Indeed, consider $x, y \in \ell_E^\pm$, and let $\psi, \phi \in \cS_E^+$ so that $\Tr(\psi) = x$ and $\Tr(\phi) = y$. Since $\cL \psi = E \psi$ and $\cL \phi = E \phi$, and since these functions are square--integrable at $+ \infty$, we can write
\begin{equation} \label{eq:isotropic}
    \omega(x,y) = \int_0^\infty [ \langle \phi, \cL \psi \rangle - \langle \cL \phi, \psi \rangle ] = 
    \int_0^\infty [ \langle \phi, E \psi \rangle - \langle \phi, E \psi \rangle ] = 0,
\end{equation}
where we used that $E \in \R$. The following result shows that if in addition $E$ is in the resolvent set of the bulk operator, then the vector spaces $\ell_E^\pm$ are Lagrangian. This result was proved in~\cite[Theorem 17]{Gon-23}. We give a short proof below in Appendix~\ref{app:proof:lagrangian_planes} for completeness.

\begin{lemma} \label{lem:lEpm_are_Lagrangian_planes} If $E \in \R \setminus \sigma(H)$, then $\ell_E^\pm$ are Lagrangian planes of $((\C^{M})^p, \omega)$, and we have
    \[
        \ell_E^+ \oplus \ell_E^- = (\C^{M})^p.
    \]
\end{lemma}

\begin{remark} \label{rem:Mp_is_even}
    If $K := Mp$ is not even, there are no Lagrangian planes. So any such operator $H$ must satisfy $\sigma(H) = \R$, and there is no gap. This happens for instance with the differential operator $\cL = - \ri \partial_t + V(t)$ with $V \in L^\infty(\R)$, so $p  = M = K = 1$. The corresponding operator $H = - \ri \partial_t + V$ is self-adjoint with domain $H^1(\R)$, and its spectrum must be the full line $\R$.
\end{remark}
In what follows, we assume $K = Mp$ to be even, and set 
\[
    N := \frac12 Mp.
\]

\subsubsection{From Lagrangian planes to unitaries}

In order to classify Lagrangian planes of $(\C^{2N}, \omega)$, we use  Leray's theorem~\cite{Ler-78} (see also~\cite[Lemma 7]{Gon-23})  which states that there is a one-to-one correspondence between such planes and unitary matrices $U \in \U(N)$. Recall that $\omega$ is described by a matrix $J$ satisfying $J^* = - J$. In particular, $A := -\ri J$ is self-adjoint, hence diagonalizable. This gives a natural splitting $\C^{2N} = K_+ \oplus K_-$, so that $- \ri J$ is positive definite on $K_+$, and negative definite on $K_-$. The dimension of $K_+$ is the number of positive eigenvalues of $A$. In what follows, we choose a basis of $\C^{2N}$ in which $J$ has the block matrix form
\begin{equation} \label{eq:def:J_A+_A-}
    J = \begin{pmatrix}
        \ri A_+ & 0 \\ 0 & - \ri A_-
    \end{pmatrix},
\end{equation}
where $A_+$ and $A_-$ are positive definite operators acting respectively on $K_+$ and $K_-$. Leray's theorem states that if we write $x \in \ell$ as $x = x_+ + x_-$ with $x_+ \in K_+$ and $x_- \in K_-$, then $x_-$ can be recovered from $x_+$ thanks to some unitary $U$. We state a version of Leray's theorem which is convenient for our purpose.

\begin{theorem}[Leray~\cite{Ler-78}] \label{th:Leray}
    The symplectic space $(\C^{2N}, \omega)$ has Lagrangian planes iff $\dim(K_+) = \dim(K_-) = N$. In this case, the Lagrangian planes of $(\C^{2N}, \omega)$ are in one-to-one correspondence with unitaries $U \in \U(N)$: $\ell$ is Lagrangian plane of $(\C^{2N}, \omega)$ iff there is $U \in \U(N)$ so that
    \begin{equation}
        \ell = \ell_U := \left\{ \begin{pmatrix}
            x \\ \sqrt{ A_-}^{-1} U \sqrt{A_+} x
        \end{pmatrix}, \qquad x \in \C^N
        \right\}.
    \end{equation}
\end{theorem}
We provide a short proof for completeness below in Appendix~\ref{sec:proofLeray}. Let us mention that the initial Leray's theorem construct a unitary from $K_+$ to $K_-$ (up to dilations). With our choice of basis, we have $K_+ = \begin{pmatrix}    \C^N \\ 0 \end{pmatrix}$ and $K_- = \begin{pmatrix} 0 \\ \C^N \end{pmatrix}$.

\begin{example} \label{ex:schrodinger_continuated}
	We continue Example~\ref{ex:schrodinger} and consider the simplest case where $M = 1$ and $V=0$. Then $H = - \Delta$ and $\sigma(H) = \mathbb R^+$. The matrix $J$ representing $\omega$ is of the form~\eqref{eq:def:J_A+_A-}, namely $J = \begin{pmatrix}
        \ri & 0 \\ 0 & - \ri
    \end{pmatrix}$, if we change basis and define the trace operator as $\Tr(\psi) = 2^{-1/2} (\psi(0) - \ri \psi'(0), \psi(0) + \ri \psi'(0))$. One can check that for any $E<0$, elements of $\cS_E \cap L^2(\R^+)$ are all of the form $\psi(x)= 2^{1/2} \alpha \re^{- \sqrt{| E |} x}$ with $\alpha \in \CC^M$ so that 
		$$
		\ell_E^+ = \left\lbrace \begin{pmatrix}
			\alpha ( 1 + \ri \sqrt{| E |}) \\   \alpha ( 1 - \ri \sqrt{| E |}).
		\end{pmatrix}, \ \alpha \in \CC^M \right\rbrace
        = \left\lbrace \begin{pmatrix}
            x \\   U_E x.
        \end{pmatrix}, \ x \in \CC^M \right\rbrace, \quad \text{with} \quad U_E := \frac{1 - \ri \sqrt{| E |}}{1 + \ri \sqrt{| E |}} \in \U(1).
		$$
    We insist on the fact that for each $E<0$, $\ell_E^+$ is associated to a \emph{single} unitary $U_E$, without specifying any kind of boundary condition for $H$ at $t=0$. 
\end{example}

\subsection{Main result: classification of symmetric Lagrangian planes}
Symmetries of $H$ translate into ones for the Lagrangian planes $\ell_E^\pm$, and for the corresponding unitaries. This gives a direct link between the symmetries of one-dimensional operators, and the original symmetric spaces of Cartan \cite{Car-26,Car-27}, which are subgroups of the unitary group $\U(N)$.  

\subsubsection{Symmetries for operators}

For Hamiltonians $H$, we consider  the three fundamental symmetries $T$, $C$ and $S$, which act only on internal degrees of freedom encoded in $\mathbb C^M$. Thus we identify $T$ on $\mathbb C^M$ with $\1_{L^2(\R)} \otimes T$ on $L^2(\R, \C^M) \approx L^2(\R \otimes \C^M)$, and similarly for $C$ and $S$. Recall that an operator $U$ (resp. $A$) on $\C^M$ is unitary  (resp. anti-unitary) if 
\[
    \forall x,y \in \C^M, \quad \bra Ux, Uy \ket = \bra x, y \ket, \qquad \text{resp.} \qquad \bra Ax, Ay \ket = \bra y, x \ket = \overline{\bra x, y \ket}.
\]

\medskip

\begin{definition}[Symmetries for operators $H$] ~
    \begin{itemize}
        \item a Time-reversal symmetry $T$ is an anti-unitary operator on $\C^M$ such that $T^2=\varepsilon_T \1$, with $\varepsilon_T=\pm1$ and with $HT=TH$;
        \item a Particle-hole symmetry $C$ is an anti-unitary operator on $\C^M$ such that $C^2=\varepsilon_C \1$, with $\varepsilon_C=\pm1$ and with $HC=-CH$;
        \item a Chiral symmetry $S$ is a unitary operator on $\C^M$ such that $S^2=\1$ and $HS=-SH$.
    \end{itemize}
    The symmetry $T$ is called bosonic if $\varepsilon_T = +1$, and fermionic if $\varepsilon_T = -1$, and similarly for the symmetry $C$.
\end{definition}
If both $T$ and $C$ are present, we assume that their product $S := T C$ is an $S$--symmetry. This is equivalent to requiring that $T$ and $C$ commute, or anti-commute, depending on their bosonic/fermionic nature. Indeed, the condition $S^2=1$ implies $TC=\varepsilon_C\varepsilon_T CT$. Listing all possibilities of the symmetries leads to the celebrated tenfold way of topological insulators~\cite{RyuSchFur-10}. 

\subsubsection{Symmetries for the symplectic space}

As claimed, the symmetries of $H$ translate into symmetries for the symplectic space $((\C^M)^p, \omega)$.

\begin{lemma}[Symmetries for symplectic form $\omega$]
    Assume that $T$, $C$ and/or $S$ are (anti-) unitary operators on $\C^M$ which are respectively $T$, $C$ and $S$ symmetries for $H$. Then, setting $T := T^{\otimes p}$ on $(\C^M)^p$ and similarly for $C$ and $S$, the corresponding symplectic form $\omega(\cdot, \cdot)$ defined in~\eqref{eq:def:omega} satisfies that for all $x, y \in (\C^M)^p$,
    \begin{equation} \label{eq:TCS_omega}
        \boxed{ \omega(Tx, Ty) = - \omega(y,x), \qquad \omega(Cx, Cy) = \omega(y,x), \qquad \omega(Sx, Sy) = -\omega(x,y). }
    \end{equation}
    In terms of the operator $J$ describing the symplectic form in \eqref{eq:defJ}, this is equivalent to
    \begin{equation} \label{eq:TCS_with_J}
        \boxed{ TJ = JT, \quad CJ = - JC, \quad S J = - JS.}
    \end{equation}
\end{lemma}
\begin{proof}
    For the $T$--symmetry, we have, for all $\phi, \psi \in C^\infty_0(\R)$, the pointwise identity
    \[
        \bra T \phi, H T \psi \ket_{\C^M} = \bra T \phi, T H \psi \ket_{\C^M} = \bra H \psi, \phi \ket_{\C^M},
    \]
    where we used that $TH = TH$ and that $T$ is anti-unitary. So
    \[
         \bra T \phi, H T \psi \ket_{\C^M} -  \bra H T \phi, T \psi \ket_{\C^M} = \bra H \psi, \phi \ket_{\C^M} -  \bra \psi, H \phi \ket_{\C^M}.
    \]
    Integrating on $\R^+$ and using the definition of $\omega$ in~\eqref{eq:def:omega}, we get $\omega(Tx, Ty) = - \omega(y,x)$. In particular, we have
    \[
        \bra J Tx, Ty \ket = - \bra Tx, JTy \ket = - \omega(Tx, Ty) = \omega(y,x) = \bra y, Jx \ket = \bra TJx, Ty \ket,
    \]
    where we used that $J^* = -J$, the identity $\omega(Tx, Ty) = -\omega(y,x)$ and the fact that $T$ is anti-unitary. We deduce as wanted that $TJ = JT$. The proofs are similar for the $C$-- and $S$-- symmetries.
\end{proof}

\begin{definition}
We say that a Lagrangian plane $\ell$ respects the $T$, $C$ and/or $S$ symmetries if $T \ell = \ell$, $C \ell = \ell$ and/or $S \ell = \ell$.
\end{definition} 

At this point we can forget about the original Hamiltonian formalism, and define the $T$, $C$ and $S$ symmetries directly at the level of some abstract symplectic space $(\cH, \omega)$, using~\eqref{eq:TCS_omega} as a definition. Our first Theorem is a classification of Lagrangian planes respecting these symmetries (without any reference to the underlying Hamiltonian).

\subsubsection{Classification of symmetric Lagrangian planes}

Recall that Cartan~\cite{Car-26,Car-27} identified ten subgroups of $\U(N)$, which are all subsets of $\U(N)$ respecting some symmetries. In the context of topological insulators, these Cartan classes are also called Altland-Zirnbauer classes~\cite{AltZir-97}. 

\begin{theorem}\label{thm:classification}
    Let $(\C^{2N}, \omega)$ be a symplectic form admitting Lagrangian planes. For each of the ten symmetry classes labelled by the Cartan's label $\Gamma$, let $\Lambda_\Gamma$ be the set of Lagrangian planes respecting these symmetries, and $U_\Gamma$ the set of unitaries obtained from them via Leray's theorem~\ref{th:Leray}. Then $U_\Gamma$ is the corresponding Cartan symmetric space.
\end{theorem}

\begin{table}[htb] \label{table:main}
    \begin{tabular}{|c|c|c|c||c||c|}
        \hline
        Cartan label & $T$ & $C$ & $S$ & Classifying space &  $\Index^\Gamma[U]$\\
        \hline
        \hline
        \hyperref[sec:class_A]{A} & 0 & 0 & 0 &  $\U(N)$ & $0$ \\
        \hline
        \hyperref[sec:class_AIII]{AIII} & 0 & 0 & 1 & $\bigcup_{k=0}^N \U(N) / \U(k) \times \U(N -k)$ & $\dim \ker (U-1) \in \{0,\ldots, N\}$ \\
        \hline
        \hline
        \hyperref[sec:class_AI]{AI} & 1 & 0 & 0  & $\U(N)/\O(N)$ & 0 \\
        \hline 
        \hyperref[sec:class_BDI]{BDI} & 1 & 1 & 1 & $\bigcup_{k=0}^N \O(N) / \O(k) \times \O(N-k)$ & $\dim \ker (U-1) \in \{0,\ldots, N\}$ \\
        \hline 
        \hyperref[sec:class_D]{D} & 0 & 1 & 0 & $\O(N)$ & $\det (U) \in \{\pm 1\}$ \\
        \hline
        \hyperref[sec:class_DIII]{DIII} & -1 & 1 & 1 & $\O(2n)/ \U(n)$ & $\mathrm{Pf} (U) \in \{\pm 1\}$ \\
        \hline
        \hyperref[sec:class_AII]{AII} & -1 & 0 & 0 & $\U(2n) / \Sp(n)$ & 0 \\
        \hline 
        \hyperref[sec:class_CII]{CII} & -1 & -1 & 1 & $ \bigcup_{k=0}^n \Sp(n)/\Sp(k) \times \Sp(n-k)$ & $\dim \ker (U-1) \in \{0,\ldots, N\}$\\
        \hline
        \hyperref[sec:class_C]{C} & 0 & -1 & 0 & $ \Sp(n)$ & 0 \\
        \hline
        \hyperref[sec:class_CI]{CI} & 1 & -1 & 1 & $\Sp(n)/\U(n)$ & 0 \\
        \hline
    \end{tabular}
    \caption{Classification of Lagrangian planes satisfying the $T$, $C$ and/or $S$ symmetries ($0$ when absent or $\pm1$ for its square when present). In the last five classes, we further assume that $N = 2n$ is also even. For each class $\Gamma$ we give the corresponding classifying space $U_\Gamma$ and an index $\Index^\Gamma$ which distinguishes possible multiple connected components.}
    \label{tab:classification}
\end{table}

We recall for convenience the ten Cartan's labels and the corresponding subgroups of unitaries in Table~\ref{tab:classification}. In this table, we also indicate the number of connected components, and an index which labels these connected components. This index has the property that $\Index^\Gamma(U_1) = \Index^\Gamma(U_2)$ iff $U_1$ and $U_2$ are in the same connected components of $U_\Gamma$. 
The proof of Theorem~\ref{thm:classification} is given in Section~\ref{sec:classification} below. For each class, we identify the corresponding expression for $\Index^\Gamma$.

\medskip

As we will clarify in the next section, our classification can be read in two different ways in the context of condensed matter: it classifies $d = 1$ dimensional bulk operators, and $d = 0$ dimensional edge properties. The table coincides with the usual $d=1$ column of the periodic table of topological insulators. However, along the proof we relate symmetric Lagrangian planes and their associated unitaries with the $d=0$ column of topological insulators from \cite{GonMonPer-22}, but where the symmetry is shifted by one row. This shift illustrates the celebrated Bott periodicity~\cite{Bot-56,Bot-56a}, which is an essential property of the periodic table \cite{RyuSchFur-10}. We refer to~\cite{Gir-25} for a proof of Bott periodicity using the Lagrangian formalism that we presented.


\subsection{Main results for operators}
\label{sec:intro:junction}

We now use our classification of Lagrangian planes in the context of condensed matter, using the relationship between these objects.

\subsubsection{Bulk operators}
Recall that our assumptions on the initial differential operator are very loose, so it is unclear yet what we really mean by <<bulk>> operators. In what follows, we say that $H$ is a {\bf bulk operator} (at energy $E$) if $H$ is self-adjoint and if $E \notin \sigma(H)$. The set of bulk operators has a natural topology, induced for instance by the distance
\[
    {\rm dist}(H_0, H_1) := \left\| (E - H_0)^{-1} - (E - H_1)^{-1} \right\|_{\rm op}.
\]
We say that $H_0$ and $H_1$ are path--connected if there is a path $(H_s)_{s \in [0, 1]}$ for this topology with $H_{s = 0} = H_0$ and $H_{s = 1} = H_1$. In particular, it implies that the {\em gap} does not close: $E \notin H_s$ for all $s \in [0, 1]$. We say that these operators are path connected in the Cartan class $\Gamma$ if the path can be chosen so that $H_s$ is in the symmetric class $\Gamma$ for all $s \in [0, 1]$. In particular, $H_0$ and $H_1$ are in the same symmetric class $\Gamma$.

\medskip

First, we record a simple Lemma, which connects the symmetries of the operators with the symmetries of the corresponding Lagrangian planes.
\begin{lemma}
    Let $H$ be a bulk operator at energy $E$ in the symmetric class $\Gamma$, with $E = 0$ if $\Gamma$ has a $C$ and/or $S$ symmetry. Then $\ell_E^\pm \in \Lambda_\Gamma$.
\end{lemma}

\begin{proof}
    This comes from the fact that the $T$, $C$, and $S$ symmetries only acts on the internal degrees of freedom. In particular, $\Psi$ is square integrable at $\pm \infty$ iff $T \Psi$, $C \Psi$ and/or $S \Psi$ are square integrable at $\pm \infty$. In addition, if $H \Psi = E \Psi$, we have
    \[
        H(T \Psi) = E(T \Psi), \quad H(C \Psi) = - E (C \Psi) , \quad H(S \Psi) = - E (S \Psi).
    \]
    This proves that $T \ell_E^\pm = \ell_E^\pm$, $C \ell_E^\pm = \ell_{-E}^\pm$ and $S \ell_E^\pm = \ell_{-E}^\pm$.
\end{proof}
In what follows, when $C$ or $S$ is present we shall always focus at $E=0$. Following~\cite{Gon-23}, we define the bulk/edge index of {\em any} bulk operator as follows.

\begin{definition}[Bulk/edge index]
    For a bulk operator $H$ in the Cartan class $\Gamma$, we define its bulk/edge index by
    \[
        \Index^\Gamma(H) := \Index^\Gamma(\ell_E^+) :=  \Index^\Gamma(U_E^+),
    \]
    where $U_E^+$ is the unitary constructed from $\ell_E^+$ via Leray's theorem.
\end{definition}
Although this index looks like an edge quantity, it really depends on the bulk properties of $H$. In particular, we did not impose any boundary condition at $t = 0$, and the sets $\ell_E^+$ and $\cS_E^+$ really depends on the behaviour of $H$ on the full right half line. We could also have used the full left half line. We will see below that $\Index^\Gamma(\ell_E^+)$ and $\Index^\Gamma(\ell_E^-)$ are related. With this, we can first recall the classification of bulk operators~\cite{Kit-09, RyuSchFur-10}.
\begin{proposition}[Classification of bulk operators]
    If $H_0$ and $H_1$ are bulk operators at energy $E$, which are path--connected in the Cartan class $\Gamma$. Then $\Index^\Gamma(H_0) = \Index^\Gamma(H_1)$.
\end{proposition}

\begin{proof}
    We only sketch the proof. If $H_s$ is a path connecting $H_0$ and $H_1$ in the class $\Gamma$, then $H_s$ defines a symplectic space $(\C^{Mp}, \omega_s)$ for each $s \in [0, 1]$, and Lagrangian planes $\ell_{E, s}^+$. In addition, the map $s \mapsto \omega_s$ is continuous for the topology of bilinear forms, and the map $s \mapsto \ell_{E, s}^+$ is continuous for the topology of vector spaces in $\C^{Mp}$. In particular, the corresponding unitaries $U_{E, s}^+$ given by Leray's theorem is a continuous family of unitaries in the Cartan symmetric space $U_\Gamma$, hence stays in the same connected component of $U_\Gamma$. So $U_{E,0}$ and $U_{E, 1}$ are path--connected, hence have the same index.
\end{proof}

A general proof that $s \mapsto \omega_s$ is continuous would be rather tedious, but in practice, the maps $\omega_s$ are explicit, and continuity can be checked directly. For instance, in the context of Schrödinger operators, if one changes the potential $V_0$ into $V_1$ continuously, then $\omega_s = \omega$ is independent of $V_s$ along the path, see Example~\ref{ex:schrodinger}.

\subsubsection{Junctions}

We now turn to the question whether interface modes appear at the junction of two bulk materials in the same symmetry class, but with different indices. More specifically, we consider two materials described by two differential operators $\cL^L$ and $\cL^R$ of order $p$, and consider the {\bf junction} differential operator $\cL^\sharp$ defined by
\[
    \forall \phi \in C^\infty_0(\R), \qquad (\cL^\sharp \phi)(x) := \begin{cases}
        (\cL^L \phi)(x), \quad \text{if} \quad x < 0, \\
        (\cL^R \phi)(x), \quad \text{if} \quad x \ge 0.
    \end{cases}
\]
This operators somehow describes a {\em hard truncation}, or {\em hard junction} between the left and right operators. The hard truncation has the advantage to immediately infer that, with obvious notation,
\begin{equation} \label{eq:ell_Esharp=ell_ER}
    \ell_E^{\sharp, +} = \ell_E^{R, +} \quad \text{and} \quad \ell_E^{\sharp, -} = \ell_E^{L, -}.
\end{equation}
However our framework also allows to handle continuous junctions, see Section~\ref{ssec:smooth_junctions} below.

\medskip
We make the following assumptions:
\begin{itemize}
    \item the differential operators $\cL^L$, $\cL^R$ and $\cL^\sharp$ define self-adjoint operators $H^L$, $H^R$ and $H^\sharp$ with the same domain $H^p(\R, \C^M)$; 
    \item the operators $H^L$ and $H^R$ are in the symmetry class $\Gamma$, and are bulk operators at energy $E$ (in the presence of $C$-- and/or $S$--symmetry, we take $E = 0$).
\end{itemize}
In particular, $Mp = 2N$ is even, see Remark~\ref{rem:Mp_is_even} above. The fact that $H^\sharp$ is self-adjoint is, actually, quite a strong assumption. It implies in particular that the symplectic forms $\omega_L$ and $\omega_R$ constructed from $H^L$ and $H^R$ coincide. Indeed, we have, for $\phi, \psi \in C^\infty_0(\R, \C^M)$,
\begin{align}
    0 & = \int_\R \left( \bra \phi, H^\sharp \psi \ket_{\C^M} - \bra  H^\sharp \phi, \psi \ket_{\C^M} \right)(t) \rd t \nonumber \\
    & = \int_{\R^-} \left( \bra \phi, H^L \psi \ket_{\C^M} - \bra H^L  \phi, \psi \ket_{\C^M} \right)(t) \rd t + \int_{\R^+} \left( \bra \phi, H^R \psi \ket_{\C^M} - \bra H^R \phi,  \psi \ket_{\C^M} \right)(t) \rd t \nonumber \\
    & = - \omega_L(\Tr(\phi), \Tr(\psi)) + \omega_R(\Tr(\phi), \Tr(\psi)). \label{eq:omegaL_omegaR}
\end{align}
In the last line, we used a similar computation with $H^L$ to prove that the integral on $\R^-$ equals the one on $\R^+$, up to a minus sign. In other words, the self-adjointness of $H^\sharp$ ensures that one cannot make a junction between two unrelated systems.

\medskip

We can now state our main theorem about junctions. In order to state it, we distinguish the cases of class $\rD$ and $\rDIII$ (where the index in valued in $\Z^2$--valued), and the cases of class $\rAIII$, $\rBDI$ and $\rCII$, where the index is valued in $\{ 0, \cdots, N\}$ with $N := \frac12 Mp$. In all these classes, the $C$ and/or $S$ symmetry is present, so we only focus on the energy $E = 0$.

\begin{theorem}[Bulk-boundary correspondence for junctions] \label{thm:BEC} Under the previous assumptions, we have
    \begin{itemize}
        \item {\bf (Classes $\rD$ and $\rDIII$)}. If $\Index^\Gamma(H^L) \neq \Index^\Gamma(H^R)$, then  $\dim \ \Ker (H^\sharp) \ge 1$.
        \item {\bf (Classes $\rAIII$, $\rBDI$ and $\rCII$)}. $\dim \ \Ker (H^\sharp) \ge \left| \Index^\Gamma(H^R) - \Index^\Gamma(H^L) \right|$.
    \end{itemize}
\end{theorem}

The right--hand side of the inequalities can be identified with a {\em relative} index, of the form $\Index^\Gamma(H^L, H^R)$, defined in a straightforward manner depending on the class. In many physical situations, this relative index is more relevant than the absolute one from the previous sections. For example, the index value for chiral chains may depend on the choice of unit cell. By nature, the index for class $\rAIII$ is only relative, and can only compare two models \cite{ProSch-16}. More recently, it has been shown that, in dimension $d=2$ and for the classes $\rA$ or $\rD$, an absolute index may not exist for some operators with unbounded spectra, and only a relative index between two operators makes sense, see~\cite{GrafJudTau-21,Bal-22,RosTar-24}.

\medskip

Theorem~\ref{thm:BEC} states that if $H^L$ and $H^R$ have different indices, then the junction between $H^L$ to $H^R$ must have at least this relative index number of zero modes at the junction. We call these modes the {\bf protected edge modes}. Note that the inequality is large: the junction Hamiltonian $H^\sharp$ may have additional (unprotected) zero modes. 

\medskip

Theorem~\ref{thm:BEC}  is proved in Section~\ref{sec:proof:junctions}. The main tools that we use are summed up in the following Lemma.

\newpage

\begin{lemma} \label{lem:crossing_lagrangian_with_unitaries} ~
	\begin{enumerate}
		\item Let $\ell_A$ and $\ell_B$ be two Lagrangian planes in $(\C^{2N}, \omega)$, with corresponding unitaries $U_A$ and $U_B$ in $\U(N)$. Then
		\[
		\dim (\ell_A \cap \ell_B) = \dim \, \Ker (U_A U_B^* - 1).
		\]
		\item With the same notation as before, and with obvious notation, we have
		\[
		\dim \left( H^\sharp - E \right) = \dim \left( \ell^{R, +}_E \cap \ell^{L, -}_E \right) =  \dim \, \Ker (U_E^{R, +} \left(U_E^{L, -} \right)^* - 1).
		\]
	\end{enumerate}
\end{lemma}

\begin{proof}
	The first point is standard, and is proved {\em e.g.} in~\cite[Lemma 10]{Gon-23}. For the second point, we recall from~\eqref{eq:def:ell_E} that 
	\[
	\dim \, \Ker (H^\sharp - E) =  \dim (\ell_E^\sharp) = \dim \left( \ell_E^{\sharp, +} \cap \ell_E^{\sharp, -} \right) =  \dim \left( \ell^{R, +}_E \cap \ell^{L, -}_E \right).
	\]
	where we used~\eqref{eq:ell_Esharp=ell_ER} in the last equality.
\end{proof}

In other words, one can read the multiplicity of $E$ as an eigenvalue of the junction operator $H^\sharp$ solely from the crossings of two Lagrangian planes depending only on the bulk operators $H^R$ and $H^L$.

\subsection{Examples}
\label{sec:examples}

Before turning to the proofs, let us give some examples of operators where our theory applies. We have already seen the case of Schrödinger operators in Example~\ref{ex:schrodinger}.

\subsubsection{Dirac operators}

Consider a differential operator of order $p = 1$, of the form
\[
\cL  = ( - \ri \partial_t) (\sigma_3 \otimes \bbI_{N}) + V(t), \quad \text{with} \quad
\sigma_3 \otimes \bbI_{N} = \begin{pmatrix}
    1 & 0 \\
    0 & -1
\end{pmatrix},
\]
where $V(t)$ is pointwise a $(2N) \times (2N)$ hermitian matrix, uniformly bounded for $t \in \R$. This defines a self-adjoint operator $\Dirac$ with domain $H^1(\R, \C^{2N})$. For all $\phi, \psi \in C^\infty_0(\R, \C^{2N})$, the integration formula gives
\[
    \int_0^\infty \left[ \bra \phi, \Dirac \psi \ket_{\C^N} - \bra \Dirac \phi, \psi \ket_{\C^N} \right](t) \rd t  = \ri \bra \phi(0), \sigma_3 \psi(0) \ket_{\C^N}.
\]
This only depends on the value of $\phi$ and $\psi$ at $0$, and not on the choices of $\phi$ and $\psi$, and defines the symplectic form $\omega : \C^{2N} \times \C^{2N} \to \C$ given by
\[
\omega(x, y) = \bra x, J y \ket_{\C^{2N}}, \quad \text{with} \quad 
J := \begin{pmatrix}
    \ri & 0 \\ 0 & - \ri
\end{pmatrix} = \ri \sigma_3 \otimes \bbI_N.
\]
Note that, as in the Schrödinger case, this symplectic form is independent of the choices of the potential $V$. 

\subsubsection{Tight--binding models}

We consider  tight-binding Schrödinger--like operators $h$ acting on $\ell^2(\Z, \C^N)$, of the form
\[
    \forall n \in \Z, \qquad (h \psi)_n = a_{n-1}^* \psi_{n-1} + b_n \psi_n + a_n \psi_{n+1},
\]
where $a_n$ and $b_n$ are $N \times N$ matrices, and $b_n$ is self-adjoint. The edge modes of such models have been recently studied in~\cite{GomGonVan-24}. In the case where the sequences $(a_n)$ and $(b_n)$ are bounded, the operator $h$ is bounded self-adjoint with domain $\ell^2(\Z, \C^N)$. 

\medskip

For our theory to apply, we further need to assume that all the matrices $(a_n)$ are invertible. Then, the equation $h \psi = E \psi$ can be seen as the linear recurrent sequence of order $2$, explicitly
\[
\forall n \in \Z, \quad \psi_{n+1} = - a_n^{-1} \left[  a_{n-1}^* \psi_{n-1} + (b_n - E) \psi_n \right].
\]
In some sense, this equation is the analogue of Cauchy's theory for ordinary differential operators. The set of solutions of $h \psi = E \psi$ is of dimension $2N$, and can be parametrized by the initial values $(\psi_0, \psi_1) \in (\C^N)^2$. We therefore define the evaluation map as $\Tr(\psi) = (\psi_0, \psi_1)^T \in \C^{2N}$ in this case (compare with~\eqref{eq:def:Tr}).

\medskip

For all $\phi, \psi \in \ell^2(\Z, \C^N)$ with compact support, we have
\begin{align*}
    \sum_{n=1}^\infty \left[ \bra \phi_n, (h \psi)_n \ket_{\C^N} - \bra (h \phi)_n,  \psi_n \ket_{\C^N} \right] 
= \bra \phi_1, a_0^* \psi_0 \ket_{\C^N} -  \bra \phi_0, a_0 \psi_1 \ket_{\C^N}.
\end{align*}
We recognize the symplectic form
\[
\omega(x, y) = \bra x, J y \ket_{\C^{2N}}, \quad \text{with} \quad J = \begin{pmatrix}
    0 & -a_0 \\ a_0^* & 0
\end{pmatrix}.
\]
Note that the sesquilinear form $\omega$ is non--degenerate since $a_0$ is invertible. Again, it is independent of the choice of $(b_n)_{n \in \Z}$ and of $(a_n)_{n \in \Z \setminus \{ 0 \}}$.

\subsubsection{Non--homogeneous media}

One can also consider operators of order $p = 2$, of the form
\[
    \cL := - \partial_t \left( A(t) \partial_t \, \cdot \right) + V(t).
\]
We assume that $A(t)$ and $V(t)$ are pointwise hermitian $N \times N$ matrices, uniformly bounded for $t \in \R$, and that there are constants $0 < \alpha \le \beta < \infty$ so that
\[
    \forall t \in \R, \qquad \alpha \bbI_M \le A(t) \le \beta \bbI_M.
\]
We recover the case Schrödinger in the case where $A(t) = \bbI_M$ pointwise. This general case has some subtleties that we would like to emphasize. First, the differential operator $\cL$ defines a symmetric operator $H$ on $L^2(\R, \C^M)$ if and only if the map $t \mapsto A(t)$ is {\bf continuous}. This is because the distributional derivative of a discontinuous function involves Dirac mass measures, which are not $L^2$ functions. Assume for instance that $A(\cdot)$ is continuous on $\R^-$ and on $\R^+$, but not necessarily at $t = 0$. Then, for $\phi, \psi \in C^\infty_0(\R)$, an integration by part shows that
\[
    \bra \phi, \cL \psi \ket_{L^2} -  \bra \cL \phi, \psi \ket_{L^2} = \omega_+(\Tr(\phi), \Tr(\psi)) - \omega_-(\Tr(\phi), \Tr(\psi)), 
\]
with $\omega_\pm(x, y) = \bra x, J_\pm y \ket_{\C^M}$ where
\[
    J_+ = \begin{pmatrix}
        0 & A(0^+) \\ - A^*(0^+) & 0
    \end{pmatrix}, \qquad 
    J_- = \begin{pmatrix}
        0 & A(0^-) \\ - A^*(0^-) & 0
    \end{pmatrix}.
\]
In the continuous case, we have $A(0^+) = A(0^-)$ and $H := \cL$ is symmetric. In this case, it defines a boundary symplectic space with the symplectic form $\omega = \omega_+ = \omega_-$. Note that this form depends on the location of the cut (here at $t = 0$). One could also cut at any $t \in \R$, and define the corresponding symplectic forms $\omega_t$. By continuity of $A(\cdot)$, the map $t \mapsto \omega_t$ is continuous.

\medskip

When considering junctions between two materials, one needs to assume that $A^L(0^-) = A^R(0^+)$, so that the junction operator $H^\sharp$ also defines a self-adjoint operator, see also~\eqref{eq:omegaL_omegaR}.

\medskip

Let us finally remark that the equation $\cL \psi = E \psi$ can be recast as 
\[
    \begin{pmatrix}
        \psi \\ 
        A \psi'
    \end{pmatrix}' = \begin{pmatrix}
        0 & A^{-1}(t) \\
        E - V(t) & 0
    \end{pmatrix} \begin{pmatrix}
    \psi \\ 
    A \psi'
    \end{pmatrix}.
\]
The corresponding initial value problem is always well-posed, thanks to the invertibility condition $A(t) \ge \alpha \bbI_M$, and whenever the map $A(\cdot)$ is continuous (and not necessarily differentiable).

\subsubsection{Other operators} We finally mention that our theory also applies {\em mutatis mutandis} for operators of the form 
\[
    \cL := R(t)^* T(- \ri \partial_t) R(t) + V(t),
\] 
where $T(k)$ is a polynomial of degree $p$ from $\R$ to the set of $M \times M$ matrices, satisfying $T(k) = T^*(k)$ pointwise. The operator $T( - \ri \partial_t)$ is defined via spectral calculus, and means that
\[
    \cF \left[ T( - \ri \partial_t) \psi\right](k) = T(k) \cF[\psi](k)
\]
where $\cF : L^2(\R) \to L^2(\R)$ is the usual Fourier transform. Then, $\cL$ is a differential operator of order $p$, and defines a self-adjoint operator on $H^p(\R, \C^M)$ whenever $V : t \mapsto \cS_N(\C)$ and $R : t \mapsto \cM_M(\C)$ are uniformly bounded in $t \in \R$, and $R(\cdot)$ is of class $C^{p-1}$, pointwise invertible with inverse uniformly bounded on $\R^d$. In some models, $R$ allows to implement an inhomogeneous velocity propagation in the material. 

The Schrödinger case corresponds to $T(k) = k^2 \bbI_M$ and $R(t) = \bbI_M$, while the Dirac case corresponds to $T(k) = k \sigma_3$ and $R(t) = \bbI_M$. Our general message here is that the theory applies to a large number of operators. We record that in the case where $T(k) = k \sigma_3$ and $R(t)$ is continuous, we have $\omega(x, y) = \bra x, J y \ket$ with $J = \ri R^*(0) \sigma_3 R(0)$.


\section{Proof of the classification}
\label{sec:classification}

In this section, we prove Theorem~\ref{thm:classification}, and identify how the symmetries of some (abstract) Lagrangian planes translate into symmetries for the corresponding unitaries. The strategy of the proof follows the lines of~\cite{GonMonPer-22}. For each class, we find a basis of $\C^{2N}$ in which the representation of the $T$, $C$, and $S$ symmetries have simple forms (we refer to~\cite{GonMonPer-22} for the construction of them). In these bases, we identify directly the constraints on the unitary $U$. 

\medskip

In what follows, we restrict our attention to the case where there is an orthonormal basis of $\C^{2N}$ in which $J$ has the form
\begin{equation} \label{eq:def:J0}
    J = \begin{pmatrix}
    \ri & 0 \\ 0 & -\ri
\end{pmatrix}, \quad \text{and} \quad \ell_U = \left\{ \begin{pmatrix} x  \\ U x\end{pmatrix}, \ x \in \C^N \right\}.
\end{equation}
The general case can always be brought to this case by modifying the inner product of $\C^{2N}$. Namely, with the notation of~\eqref{eq:def:J_A+_A-}, we can define
\[
    \bra x, y \ket_{J} :=  \bra x_+, A_+ x_+ \ket_{\C^N} + \bra x_-, A_- x_- \ket_{\C^N}.
\]
The fact that $A_+$ and $A_-$ are positive shows that $\bra \cdot, \cdot \ket_J$ is a scalar product, and, relative to this inner product, the new matrix $J$ representing $\omega(\cdot, \cdot)$ is of the form~\eqref{eq:def:J0}.

\medskip

Here and thereafter, we denote by $\cK : \C^k \mapsto  \C^k$ the usual complex-conjugation operator. Recall that for a matrix $M \in \cM_n(\mathbb C)$, we denote its conjugate matrix by $\overline{M} := \cK M \cK$ and its transpose by $M^T := \overline{M}^*$. 

\subsection{Class A}
\label{sec:class_A}

This class has no symmetry at all. Leray's theorem directly implies that we recover the whole unitary group $\U(N)$, which is simply connected. The index is zero.

\subsection{Class AIII}
\label{sec:class_AIII}

This class includes the $S$-symmetry only. Since $SJ+JS=0$ and $S^2 = 1$, there exists a basis of $\C^{2N}$ where
$$
S = \begin{pmatrix}
 0 & 1 \\
1 & 0
\end{pmatrix}.
$$
Let $\ell$ be a Lagrangian plane stable under the $S$--symmetry: $S \ell = \ell$, and let $U \in \U(N)$ be its corresponding unitary. This implies that for any $x \in \CC^N$, it exists $y \in \CC^N$ such that
$$
S \begin{pmatrix}	x  \\	U x \end{pmatrix} = 
\begin{pmatrix} y  \\ U y \end{pmatrix}, \qquad \text{hence} \qquad
\begin{cases}
    Ux = y \\
    x = U y
\end{cases}, \qquad \text{which gives} \qquad
U = U^*.
$$
This corresponds to unitary and self-adjoint matrices:
\[
    \boxed{ \U(N) \cap \cS_N(\C). }
\]
Such unitaries can only have $\pm 1$ in their spectrum. So they are parametrized by subspaces of rank $k \in \{0, \ldots, N\}$, with $k  = \dim \, \Ker (U - 1)$. The classifying space is therefore
$$
\cG = \mathop{\bigcup}_{k=0}^N \cG(k,N), \quad \text{where} \quad
\cG(k,N) \cong \U(N) / \U(k) \times \U(N-k)
$$
is the usual (complex) Grassmannian. Each $\cG(k,N) $ is path--connected, thus $\cG$ has $N+1$ connected components, and the corresponding index is $$\Index^{\rAIII}[U] := \dim \ker(U-1) \in \{0,\ldots, N\}.$$

\subsection{Class AI}
\label{sec:class_AI}

This class includes the $T$-symmetry only with $T^2=+1$. Since $TJ=JT$, we consider a basis where 
$$
T = \begin{pmatrix}
	0 & \cK \\
	\cK & 0
\end{pmatrix},
$$
where we recall that $\cK$ is the usual complex conjugation operator on $\C^N$.  Let $\ell$ be a Lagrangian plane stable under the $T$-symmetry: $T\ell=\ell$. Thus for $x \in \CC^N$ it exists $y \in \CC^N$ such that 
$$
T \begin{pmatrix}	x  \\	U x \end{pmatrix} = 
\begin{pmatrix}	y  \\	U y \end{pmatrix}, \qquad \text{hence} \qquad
\begin{cases}
    \overline{U} \overline{x} = y \\
    x = U y
\end{cases}, \qquad \text{which gives} \qquad
 U = U^T.
$$
So the set of unitaries one obtains are the set of real--symmetric unitaries
\[
    \boxed{ \U(N) \cap \cS_N^\R(\C)} \qquad  \cong \U(N)/\O(N).
\]
The last equivalence is proved for instance in~\cite[Corollary A.2]{GonMonPer-22}. This set is path-connected so that the index is zero.

\subsection{Class BDI}
\label{sec:class_BDI}

This class includes all three symmetries with $T^2=+1$, $C^2=+1$ and $S=CT$, so that $S^2=1$. We consider a basis where
$$
S = \begin{pmatrix} 0 & 1 \\ 1 & 0 \end{pmatrix}, \qquad 
T = \begin{pmatrix} 0 & \cK \\ \cK & 0 \end{pmatrix} , \qquad 
C=\begin{pmatrix} \cK & 0 \\ 0 & \cK \end{pmatrix}.
$$
where $\cK$ is the complex conjugation. Let $\ell$ be a Lagrangian plane invariant under the three symmetries. Since $S=CT$, we only need to look at $C$ and $T$. Like in class $\rD$ below, the $C$-symmetry implies that $\overline{U}=U$, so $U \in \O(N)$, and like in class $\rAI$ above, the $T$-symmetry implies that $U^T=U$, to $U$ is also real--symmetric. Hence, $U$ is an orthogonal symmetry, but with respect to a real subspace, which gives
\[
    \boxed{ \O(N) \cap \cS_N^\R(\C) .}
\]
The classifying space is similar to class $\rAIII$ but with orthogonal matrices, that is
$$
    \cG^\RR = \mathop{\bigcup}_{k=0}^N \cG^\RR(k,N), \quad \text{where} \quad
    \cG^\RR(k,N) \cong \O(N) / \O(k) \times \O(N-k)
$$
is the usual (real) Grassmannian, which is path-connected. Thus $\cG^\RR$ has $N+1$ connected components, and the corresponding index is 
$$
\Index^{\rBDI}[U] := \dim \ker(U-1) \in \{0,\ldots, N\}.$$

\subsection{Class D}
\label{sec:class_D}

This class includes the $C$-symmetry only with $C^2=+1$. Since $CJ=-JC$, we consider a basis where $C = \cK$. If $\ell$ with corresponding unitary $U$ is stable for $C$, then for all $x \in \CC^N$ there is $y \in \CC^N$ such that
$$
C \begin{pmatrix}	x  \\	U x \end{pmatrix} = 
\begin{pmatrix}	y  \\	U y  \end{pmatrix},
\quad \text{so} \quad
\begin{cases}
    \overline{x} = y \\
    \overline{U} \overline{y} = Uy,
\end{cases}
\quad \text{which gives} \quad
\overline{U}= U.
$$
So $U$ is real-valued, and the classifying space is 
\[
    \boxed{ \O(N). }
\]
It has two connected components, and the index is given by
$$ \Index^{\rD} [U] := \det (U) \in \{\pm 1\}.$$

\subsection{Class DIII} 
\label{sec:class_DIII}

This class includes $T$ with $T^2=-1$, $C$ with $C^2=1$ and $S=TC = - CT$. We consider a basis where
$$
T =  \begin{pmatrix}
0 & -\cK\\
\cK & 0
\end{pmatrix}, \qquad C= \ri\cK, \qquad S = \begin{pmatrix}
0 & \ri \\
-\ri & 0
\end{pmatrix}.
$$
Consider a Lagrangian plane $\ell$ with corresponding matrix $U$, and compatible with all these symmetries. Since $S=TC$, we only need to look at $C$ and $T$. Like in class $\rAII$ below, the fermionic $T$-symmetry implies $U^T=-U$, so $U \in \cA^\RR_N(\C)$ is (real) anti-symmetric. Similarly, like in class $\rD$, the $C$-symmetry implies $\overline{U}=U$, so $U \in \O(N)$ is real--valued. Thus  $U\in \O(N)\cap \cA^\RR_N(\C)$. This set is non-empty only if $N=2n$ is even, in which case we have (for all these facts, we refer {\em e.g.} to ~\cite[Cor. A4]{GonMonPer-22})
$$
    \boxed{ \O(N)\cap \cA^\RR_N(\C)} \qquad  \cong \O(2n)/\U(n).
$$
This set has two connected components and the index is 
$$ \Index^{\rDIII}[U] := \Pf(U) \in \{\pm 1\},$$ where $\Pf$ is the Pfaffian. Recall that the Pfaffian is well-defined for anti--symmetric matrices of even size.

\subsection{Class AII}
\label{sec:class_AII}

This class includes the $T$-symmetry only with $T^2=-1$. Since $TJ=JT$, we consider the basis where 
$$
T =  \begin{pmatrix}	0 & -\cK \\
	\cK & 0
\end{pmatrix}.
$$
Let $\ell$ be a Lagrangian plane with corresponding matrix $U \in \U(N)$, stable for the $T$--symmetry. Then, for $x \in \CC^N$ it exists $y \in \CC^N$ such that 
$$
T \begin{pmatrix}	x  \\	\cU x \end{pmatrix} = \begin{pmatrix}	y  \\	\cU y \end{pmatrix},
\quad \text{hence} \quad
\begin{cases}
    - \overline{U} \overline{x} = y \\
    \overline{x} = Uy
\end{cases}, \quad \text{which gives} \quad U^T = -U.
$$
This correspond to unitary and (real) anti-symmetric matrices. This set is non--empty iff $N = 2n$ is even,  in which case we have
\[
    \boxed{ \U(N) \cap \cA_N^\R(\C)} \quad  \cong \U(2n) / \Sp(n).
\]
This set is simply connected and the index is zero. We refer again to~\cite[Thm 4.7 and Cor. A.4]{GonMonPer-22} for an elementary proof of these facts. 

\subsection{Class CII} 
\label{sec:class_CII}

This class includes $T$ with $T^2=-1$, $C$ with $C^2=-1$ and $S=TC$. This requires $N=2n$ to be even, and we consider a basis where
$$
T =  - \begin{pmatrix}	0 &  \Omega \\ \Omega & 0 \end{pmatrix} \cK, \qquad 
C= \begin{pmatrix} \Omega & 0\\ 0 & \Omega \end{pmatrix}\cK, \qquad 
S = \begin{pmatrix}	0 & 1 \\ 1 & 0 \end{pmatrix}, \quad \text{where} \quad
\Omega :=  \begin{pmatrix}	0 & \bbI_n \\ - \bbI_n & 0 \end{pmatrix}
$$
is the $(2n) \times (2n)$ symplectic matrix defined in~\eqref{eq:def:symplectic_Omega}. Let $\ell$ be a Lagrangian plane with corresponding matrix $U \in \U(N)$, and compatible for all these symmetries. Since $S=TC$, we only need to look at $C$ and $S$. Like in class \hyperref[sec:class_AIII]{AIII} above, the $S$-symmetry implies $U^*= U$, so  $U \in \cS_N(\mathbb C)$ is hermitian (in particular, $\sigma(U) \in \{ \pm 1\}$), and like in class \hyperref[sec:class_C]{C} below, the $C$-symmetry implies that $U^T \Omega U = \Omega$, so $U \in \Sp(n)$ is symplectic. We prove in Appendix~\ref{app:symp_grass} below that the space $\Sp(n) \cap \cS_{2n}(\C)$ has $n+1$ connected component given by 
$$
    \boxed{ \Sp(n) \cap \cS_{2n}(\C)} \quad \cong \bigcup_{k=0}^n \Sp(n) / \Sp(k) \times \Sp(n-k),
$$
where the right-hand side is the symplectic Grassmanian. Each component is simply connected, and the index is $$\Index^{\rCII}[U] := \dim \, \Ker \, (U-1).$$

\subsection{Class C}
\label{sec:class_C}

This class includes the $C$-symmetry only with $C^2=-1$. Since $CJ=-JC$, we get that $N=2n$ is even, and there a basis in which
$$
C =  \begin{pmatrix}	\Omega  & 0 \\	0& \Omega \end{pmatrix} \cK,
$$
with $\Omega$ the symplectic matrix given in~\eqref{eq:def:symplectic_Omega}. Let $\ell$ be a Lagrangian plane with corresponding unitary $U \in \U(N)$, stable under $C$. For all $x \in \CC^N$, there is $y \in \CC^N$ such that 
$$
C \begin{pmatrix} x  \\	U x  \end{pmatrix} = 
\begin{pmatrix}	y  \\	U y \end{pmatrix}, 
\quad \text{hence} \quad
\begin{cases}
    \Omega \overline{x} = y \\
    \Omega \overline{U} \overline{x} = U y,
\end{cases}
\quad \text{so that} \quad
    U^T \Omega U = \Omega.
$$
In other words, $U$ is symplectic. Thus the classifying space is the set of unitary symplectic matrices, namely
$$
    \boxed{ \U(N) \cap \Sp(N,\CC)} \quad =: \Sp(n),
$$
which is simply connected. The index is zero.

\subsection{Class CI} 
\label{sec:class_CI}

This class includes $T$ with $T^2=1$, $C$ with $C^2=-1$ and $S=TC$. This requires $N=2n$ to be even, and we consider a basis where
$$
T =  \begin{pmatrix} 0 & \ri \cK \\ 	\ri \cK & 0 \end{pmatrix}, \qquad 
C= \begin{pmatrix}	\Omega \cK & 0\\ 	0 & \Omega \cK \end{pmatrix}, \qquad 
S = \begin{pmatrix}	0 & \ri \Omega\\	\ri \Omega & 0\end{pmatrix}.
$$
with $\Omega$ the symplectic matrix given in~\eqref{eq:def:symplectic_Omega}. Let $\ell$ be Lagrangian plane with corresponding matrix $U \in \U(N)$, and satisfying these symmetries. Since $S=TC$, we only need to look at $C$ and $T$.  Like in class \hyperref[sec:class_AI]{AI} above, the $T$-symmetry implies that $U \in \cS^\RR_N(\CC)$, and like in class \hyperref[sec:class_C]{C} above, the $C$-symmetry implies that $U \in \Sp(n)$. So the classifying space is
\[
    \boxed{ \cS^\RR_{2n}(\CC) \cap \Sp(n) } \quad   \cong \Sp(n)/\U(n),
\]
which is simply connected. The index is zero. We refer {\em e.g.} to \cite[Cor. A2]{GonMonPer-22} for these facts.

\section{Proofs for the junctions between two topological insulators}
\label{sec:proof:junctions}

We now prove Theorem~\ref{thm:BEC} about the existence of protected modes in junctions, proceeding class by class. We only focus on the classes $\rAIII$, $\rBDI$, $\rCII$ (where the index is a dimension of a kernel), $\rD$ (where the index is a determinant) and $\rDIII$ (where the index is a Pfaffian). In all these cases, we have a $C$ or $S$ symmetry, so we take $E = 0$.

\subsection{Classes AIII, BDI, CII}
Let us first focus on the classes $\rAIII$, $\rBDI$ and $\rCII$, whose classifying spaces are respectively
\[
    \U(N) \cap \cS_N(\C), \quad \O(N) \cap \cS_N(\C), \quad \text{and} \quad \Sp(n) \cap \cS_{2n}(\C).
\]
This constraints the matrix $U$ to satisfy $\sigma(U) \in \{ \pm 1\}$. In what follows, we focus on the $\rAIII$ case for the sake of clarity, but the argument is similar for the classes $\rBDI$ and $\rCII$, up to replacing the field $\C$ by $\R$ or $\mathbb{H}$ (quaternions). Recall that
\[
     \Index^{\rAIII} [U] = \dim \, \Ker(U - 1)
\]
is a dimension of a kernel. The arguments we use follows the ones in~\cite{Gom-}. First we record the following. 
\begin{lemma} \label{lem:crossing_lagrangian_AIII}
    Let $\ell_A$ and $\ell_B$ be Lagrangian planes in the $\rAIII$ class, and let $U_A, U_B \in \U(N)$ be the corresponding unitaries. We have
    \[
        \dim (\ell_A \cap \ell_B) = \dim \, \Ker (U_A^* U_B - 1) \ge \left| N - \Index^\rAIII[U_A] -  \Index^\rAIII[U_B] \right|.
    \]
\end{lemma}

\begin{proof}
    The first equality is the first point of Lemma~\ref{lem:crossing_lagrangian_with_unitaries}. Let us prove the second inequality. We write
    \[
        \C^{N} = E_{A, 1} \oplus E_{A,-1} \qquad \text{and} \qquad \C^N = E_{B, 1} \oplus E_{B, -1},
    \]
    with $E_{A, 1} = {\rm Ker}(U_A - 1)$, $E_{A, -1} = {\rm Ker}(U_A + 1)$, and so on. If $U_A v = \pm v$ and $U_B v = \pm v$, then $U_A^* U_B v = v$, so
    \[
        \left( E_{A, 1} \cap E_{B, 1} \right) \oplus \left(  E_{A, -1} \cap E_{B, -1} \right) \subset \Ker (U_A^* U_B - 1).
    \]
    Recall that, for any vector subspaces $G, H$ of $\C^N$, we have $\dim(G + H) \le N$, from which we infer 
    \[
        \dim (G \cap H) = \dim(G) + \dim(H) - \dim(G + H) \ge \dim(G) + \dim (H) - N.
    \]
    Together with the fact that $\dim(G \cap H) \ge 0$, we get
    \[
        \dim (G \cap H)  \ge \max \left\{ 0, \dim(G) + \dim (H) - N \right\}.
    \]
    Similarly, we have $\dim(G^\perp \cap H^\perp)  = \max \{ 0, N - \dim(G) - \dim(H) \}$, so
    \[
        \dim \left( ( G \cap H) \oplus (G^\perp \cap H^\perp) \right) = \dim(G \cap H) + \dim(G^\perp \cap H^\perp) \ge \left| \dim (G) + \dim(H) - N \right|.
    \]
    and the result follows.
\end{proof}

We first deduce the following corollary. Recall that we consider the energy $E = 0$.
\begin{lemma} \label{lem:ell+_ell-_bulkCase}
   Let $H$ be an operator on $L^2(\mathbb R,\C^N)$ in the class $\rAIII$, such that $0 \notin \sigma (H)$. Let $\ell_0^\pm$ be the associated Lagrangian planes at $E=0$. Then
    \[
          \Index^\rAIII(\ell_0^+) +  \Index^\rAIII( \ell_0^-) = N.
    \]
\end{lemma}
\begin{proof}
    Apply the previous Lemma together with the fact that 
    $
        0 = \dim \, \Ker(H) = \dim \left( \ell_0^+ \cap \ell_0^- \right).
    $
\end{proof}

We can now prove Theorem~\ref{thm:BEC} for class AIII.

\begin{proof}[Proof of Theorem~\ref{thm:BEC}, Class AIII] 
    Using Lemma~\ref{lem:crossing_lagrangian_with_unitaries} at $E = 0$ and Lemma~\ref{lem:crossing_lagrangian_AIII} show that
     \begin{align*}
        \dim \, \Ker \left( H^\sharp \right)  = \dim \left( \ell_0^{R, +} \cap \ell_0^{L, -}\right)
        &  \ge \left| N - \Index^\rAIII(\ell_0^{R, +}) -  \Index^\rAIII(\ell_0^{L, -} ) \right| \\
        & =  \left| \Index^\rAIII(\ell_0^{R, +}) -  \Index^\rAIII(\ell_0^{L, +} ) \right|,
    \end{align*}
    where we used Lemma~\ref{lem:ell+_ell-_bulkCase} for the last line.
\end{proof}

\subsection{Class D} We now prove Theorem~\ref{thm:BEC} for the class $\rD$, where the classifying space is $\O(N)$. For a Lagrangian plane $\ell$ with corresponding unitary $U$ we have $\Index^\rD(\ell) := \Index^\rD[U]= \det(U) \in \{ \pm 1\}$. Our goal is to prove that if $\Index^\rD(H^R) \neq \Index^\rD(H^L)$, then $0$ is an eigenvalue of $H^\sharp$ of multiplicity at least~$1$. The proof relies on the following remark.

\begin{lemma} 
    Let $U \in \O(N)$ be an orthogonal matrix. If $\det(U) = -(-1)^N$, then $1 \in \sigma(U)$ is in the spectrum of $U$.
\end{lemma}
\begin{proof}
    All eigenvalues of $U$ have modulus $1$. In addition, since the characteristic polynomial of $U$ has real coefficients, the non--real eigenvalues occur in conjugate pair $(\lambda, \overline{\lambda})$ with equal multiplicity. In particular, each such conjugate pair contributes to $+1$ to the determinant. So, if $m_{-1}$ and $m_1$ denote the multiplicity of $-1$ and $1$ respectively, we have
    \[
        \det(U) = (-1)^{m_{-1}} , \quad \text{and} \quad m_{-1} + m_1 = N \ {\rm mod} \ 2.
    \]
    In the case where $N = 2n$ is even and $\det(U) = -1$, we must have $m_{-1}$ and $m_1$ odd, and in particular, $m_1 \ge 1$. The other case is similar.
\end{proof}

The next result is similar to Lemma~\ref{lem:crossing_lagrangian_AIII}, but for the $\rD$ class.
\begin{lemma}
	Let $H$ be an operator on $L^2(\mathbb R,\C^N)$ in the $\rD$ class such that  $0 \notin \sigma (H)$. Let $\ell_0^\pm$ be the associated Lagrangian planes at $E=0$. 
     Then $\Index^\rD(\ell_0^+) = (-1)^N \Index^\rD(\ell_0^-)$.
\end{lemma}

\begin{proof}
    Since $0 \notin \sigma(H)$, we have $\dim (\ell_0^+ \cap \ell_0^-) = 0$, so $1 \notin \sigma \left( U_0^{+} (U_0^-)^*\right)$. From the previous Lemma, we get $\det( U_0^{+} (U_0^-)^* ) = (-1)^N$, which is the result.
\end{proof}

We can now prove Theorem~\ref{thm:BEC}
\begin{proof}[Proof of Theorem~\ref{thm:BEC}, Class $\rD$] This time, we write
    \[
        \dim \, \Ker(H^\sharp)  = \dim \, \Ker ( U_0^{R, +} (U_0^{L, -})^* - 1  )
    \]
    which is non null whenever
    \[
        -(-1)^N = \det (U_0^{R, +} (U_0^{L, -})^*) = \frac{\det(U_0^{R,+})}{\det(U_0^{L,-})} = (-1)^N \frac{\Index^\rD[H^R]}{\Index^\rD[H^L]},
    \]
    that is when $\Index^\rD[H^R] = -\Index^\rD[H^L]$, which proves the result.
\end{proof}

\subsection{Class DIII} Finally, we prove Theorem~\ref{thm:BEC} for the class $\rDIII$, where the classifying space is $\O(N) \cap \cA_N^{\R}(\R)$, that is the set of orthogonal antisymmetric matrices. This set is non empty only if $N = 2n$ is even. In this case, for a Lagrangian plane $\ell$ with corresponding unitary $U \in \O(N) \cap \cA_N^{\R}(\R)$ we have
\[
    \Index^\rDIII(\ell) := \Index^\rDIII[U]= \Pf (U) \quad \in \{ \pm 1\}.
\]
This time, we use the following properties of Pfaffians. 

\begin{lemma} \label{lem:Pfaffian_identity}
    Let $A, B \in \O(2n) \cap \cA_{2n}^\R(\R)$. If $\Pf(A) = - \Pf(B)$, then $1 \in \sigma(AB)$.
\end{lemma}

\begin{proof}
    We use the following formula, which states that for all $A, B \in \cA_{2n}^{\R}(\R)$, we have (see for instance~\cite{Kri-16})
    \begin{equation} \label{eq:Pfaffian_identity}
        \Pf(A) \Pf(B) = \exp \left( \frac12 \Tr \log (A^T B ) \right).
    \end{equation}
    Assume otherwise that $1 \notin \sigma(AB)$. Then $-1 \notin \sigma(-A B) = \sigma(A^T B)$ since $A^T = -A$. 
    
    In particular, we can choose the complex logarithm with the usual branch cut at $(-\infty, 0)$ in~\eqref{eq:Pfaffian_identity}. Since $A^T B \in \O(2n)$ is orthogonal, all eigenvalues are of modulus $1$. If $\lambda \in \C \setminus \R$ is such a non-real eigenvalue, then $\overline{\lambda}$ is also an eigenvalue with same multiplicity. Note that
    \[
        \log(\lambda) + \log(\overline{\lambda}) = \log(| \lambda |^2) = \log(1) = 0,
    \]
    so such pairs of eigenvalues do not contribute in the $\Tr \log(A^T B)$. Similarly, the eigenvalue $1$ does not contribute, and since $-1$ is not in the spectrum of $A^T B$, we obtain
    \[
        \Tr \log(A^T B) = 0, \quad \text{so} \quad \Pf(A) \Pf(B) = 1.
    \]
\end{proof}

The analogue of Lemma~\ref{lem:crossing_lagrangian_AIII} for the $\rDIII$ class reads as follows. Recall that $\Pf(A^T) = (-1)^n \Pf(A)$.
\begin{lemma}
    Let $H$ be an operator on $L^2(\mathbb R,\C^{2n})$ in the $\rDIII$ class such that  $0 \notin \sigma (H)$. Let $\ell_0^\pm$ be the associated Lagrangian planes at $E=0$. 
    Then $\Index^\rDIII(\ell_0^+) = (-1)^n \Index^\rDIII(\ell_0^-)$.
\end{lemma}

\begin{proof}
    We have $0 = \dim \, \Ker H = \dim (\ell_0^+ \cap \ell_0^- ) = \dim \, \Ker(U_0^+ (U_0^-)^T - 1)$, so $1 \notin (U_0^+ (U_0^-)^T)$. Since $U_0^{\pm} \in \O(2n) \cap \cA_{2n}^\R(\R)$, we can apply the previous Lemma, and deduce that
    \[
        \Pf(U_0^+) =  \Pf((U_0^-)^T) = (-1)^n \Pf(U_0^-).
    \]
\end{proof}

We can now prove Theorem~\ref{thm:BEC} for the $\rDIII$ class.

\begin{proof}[Proof of Theorem~\ref{thm:BEC}, class $\rDIII$]
    We have 
    \[
        \dim \, \Ker (H^\sharp) = \dim (\ell_0^{R, +} \cap \ell_0^{L, -}) = \dim \, \Ker \left( U_0^{R, +} (U_0^{L, -})^* - 1 \right).
    \]
    In addition, we have $\Pf(U_0^{R, +}) = \Index^\rDIII(H^R)$ and 
    \[
        \Pf((U_0^{L, -})^T) = (-1)^n \Pf(U_0^{L, -}) = \Pf(U_0^{L, +}) = \Index^\rDIII(H^L).
    \]
    In the case where the two indices differ, we get from Lemma~\ref{lem:Pfaffian_identity} that $1 \in \sigma(U_0^{R, +}  (U_0^{L, -})^T)$, hence $\dim \, \Ker (H^\sharp) \ge 1$.
\end{proof}

\subsection{The case of continuous junctions}
\label{ssec:smooth_junctions}

In this section, we explain how to modify the proofs in the case of continuous junctions. First, we note that the notion of continuous junctions only makes sense for one--dimensional models set on $\R$ (instead of tight--binding models on $\Z$).

\medskip

We consider a junction operator described by a differential operator $\cL^\sharp$ with {\em continuous} coefficients, and so that
\[
    \cL^\sharp \psi = \begin{cases}
         \cL^L \psi, & \quad \text{for} \quad \psi \in C^\infty_0((-\infty, -X]), \\
        \cL^R \psi, & \quad \text{for} \quad  \psi \in C^\infty_0([X, \infty)),
    \end{cases}
\]
for some $X > 0$ large enough. This means that $\cL$ behaves as $\cL^L$ on the far left, and as $\cL^R$ on the far right, hence models a continuous junction between two bulk media. A typical example is
\begin{equation} \label{eq:ex:Schrodinger_junction}
    \cL(t) = - \partial_{tt}^2 + V^\sharp(t), \qquad V^\sharp = V^L(t) (1 - \chi(t)) + V^R(t) \chi(t),
\end{equation}
where $V^L$ and $V^R$ are bulk potentials (think of periodic functions), and $\chi$ is a continuous cut-off functions with $\chi(t) = 0$ for $t < -X$ and $\chi(t) = 1$ for $t > X$. We then make the same assumption as in Section~\ref{sec:intro:junction}, namely that these differential operators define self-adjoint operators $H^L$, $H^R$, and $H^\sharp$, and that $E \notin \sigma(H^{\rm L/R})$.

\medskip

The main result of this section is the following. 
\begin{lemma}\label{lem:continuous_junction}
    For the junction operator $H^\sharp$, the vector spaces $\ell_E^{\sharp, +}$ and $\ell_E^{\sharp, -}$ are Lagrangian planes. In addition, if $H^L$, $H^R$ and $H^\sharp$ are in the same symmetry class $\Gamma$, then
    \[
        \Index^\Gamma(\ell_E^{\sharp, +}) = \Index^\Gamma(\ell_E^{R, +}), \quad \text{and} \quad
        \Index^\Gamma(\ell_E^{\sharp, -}) = \Index^\Gamma(\ell_E^{L, -}).
    \]
\end{lemma}
In the previous hard cut case, we had the equality $\ell^{\sharp, +}_E = \ell^{R, +}_E$ and  $\ell^{\sharp, -}_E = \ell^{L, -}_E$. Now, these spaces are different in general, and only the indices are equal. This result {\em does not} follow from Lemma~\ref{lem:lEpm_are_Lagrangian_planes}, since we expect $E$ to be the spectrum of $H^\sharp$. In particular, we warn that the equality
\[
    \ell_E^{\sharp, +} \oplus \ell_E^{\sharp, +} = \C^{Mp}
\]
is {\bf false} in general: we expect these planes to cross, and any initial condition in this crossing gives an edge state. 

\begin{proof}
    The main idea of the proof is to consider the vectorial planes $\ell_E^{\sharp, \pm}(t)$ parametrized by $t \in \R$, corresponding  to the boundary values of $\psi \in \cS_E^{\sharp, \pm}$ at $t \in \R$. The previous case corresponds to $t = 0$.
    
    \medskip
    
    All the statements for $t= 0$ translates directly for any $t \in \R$. In particular, we obtain a family of symplectic forms $\omega_t$ on $\C^{Mp}$, for all $t \in \R$. For this symplectic form, the sets $\ell_E^{\sharp, \pm}(t)$ are still isotropic, see~\eqref{eq:isotropic}. In particular, we have 
    \[
        \dim \left( \ell_E^{\sharp, \pm}(t) \right) \le N, \quad \text{with} \quad N := \frac{Mp}{2}.
    \]    
    Cauchy's theory shows that $\dim \left( \ell_E^{\sharp, \pm}(t) \right) = \dim \, \cS_E^{\sharp, \pm}$ are independent of $t$. But for $t > X$, we recover the right bulk Hamiltonian, so $\dim \, \cS_E^{\sharp, +} = \dim \, \cS_E^{R, +}$, and similarly, $\dim \, \cS_E^{\sharp, -} = \dim \, \cS_E^{L, -}$. Finally, according to Lemma~\ref{lem:lEpm_are_Lagrangian_planes} applied for the bulk Hamiltonians $H^L$ and $H^R$, these vector spaces are of dimension $N$. So 
    \[
        \dim \left( \ell_E^{\sharp, \pm}(t) \right) = N, \quad \text{for all} \quad t \in \R.
    \]
    We proved that the sets $\ell_E^{\sharp, \pm}(t)$ are maximally isotropic, hence are Lagrangian. This proves the first part of the theorem.
    
    \medskip
    
    We now prove equality of the index. Note that since $\ell_E^{\sharp, \pm}(t)$ are Lagrangian for $\omega_t$, one can associate a unitary $U_E^{\sharp, \pm}(t) \in \U(N)$ thanks to Leray's theorem. The continuity of $\cL$ implies that $t \mapsto \omega_t$ is continuous, and that $t \mapsto \ell_E^{\sharp, \pm}(t)$ are continuous (for the topology of vector spaces of $\C^{Mp}$). In particular, the corresponding maps of unitaries $t \mapsto U_E^{\sharp, \pm}(t)$ are also continuous in $\U(N)$. Since the index is constant on connected components, we directly get
    \[
        \Index^\Gamma( U_E^{\sharp, +}(t = 0) ) = 
        \Index^\Gamma( U_E^{\sharp, +}(t = X) )  = 
        \Index^\Gamma( U_E^{R, +}(t = X) ) = \Index^\Gamma( U_E^{R, +}(t = 0) ),
    \]
    so $ \Index^\Gamma(\ell_E^{\sharp, +}) = \Index^\Gamma(\ell_E^{R, +})$ as claimed. The proof on the other side is similar.
\end{proof}
Notice that the continuity of $t \mapsto \omega_t$ can be proved directly for a given model. For instance, in the context of Schrödinger operator in~\eqref{eq:ex:Schrodinger_junction}, $\omega_t$ is independent of $t$ (so continuous). Note that the continuity of $t \mapsto \ell_E^{\sharp, +}(t)$ follows from Cauchy's theory: the functions $\psi \in \cS_E^+$ are of class $C^p(\R)$ so the map $t \mapsto (\psi(t), \psi'(t),\cdots \psi^{p-1}(t))$ is continuous.

\medskip

Theorem~\ref{thm:BEC} now applies for any smooth junctions between two bulk materials, by writing (here for the case $\Gamma \in \{ \rAIII, \rBDI, \rCII \}$)
\begin{align*}
    \dim \, \Ker (H^\sharp) & \ge \left| \Index^\Gamma( \ell_E^{\sharp, +}) - \Index^\Gamma( \ell_E^{\sharp, -}) \right| \\
    & = \left| \Index^\Gamma( \ell_E^{R, +}) - \Index^\Gamma( \ell_E^{L, -}) \right|
     = \left| \Index^\Gamma(H^R) - \Index^\Gamma(H^L) \right|.
\end{align*}

\section{Applications: Bulk--edge index for Dirac operators}
\label{sec:application_cst_Dirac}

\subsection{Dirac operators with constant potential}
In this section, we compute the index of  Dirac operators with {\em constant} potential, acting on $L^2(\R, \C^N)$ with block form
\[
\Dirac = \begin{pmatrix}
    - \ri \partial_t & - \ri W \\ \ri W^* & \ri \partial_t
\end{pmatrix}
\]
where $W \in \cM_{2N}(\C)$ is a constant matrix, independent of $t \in \R$. The $\ri$ factor in front of $W$ is here to simplify the computations below. This example highlights many interesting features of the previous results. In particular, each symmetry class of the table and the corresponding index can be directly deduced from $W$ and its spectral properties. Our results are gathered in the next Proposition. 

\begin{proposition}
    The operator $\Dirac$ acting on $L^2(\R, \C^{2N})$ is essentially self-adjoint, with domain $H^1(\R, \C^{2N})$. Its spectrum is purely essential, given by
    \[
    \sigma (\Dirac) = \sigma_{\rm ess}(\Dirac) =  (- \infty, m_0] \cup [ m_0, \infty),
    \]
    where $m_0$ is the lowest singular value of $W$. In particular, $0 \notin \sigma(\Dirac)$ iff $W \in {\rm GL}_N(\C)$ is invertible. In this case, we can identify the Lagrangian planes $\ell_{E = 0}^\pm$ as
    \[
        \ell_0^+ = \Ran \left(  \1( A > 0 ) \right) , \quad \ell_0^- =  \Ran \left( \1(A < 0) \right) , \qquad \text{with} \quad A := \begin{pmatrix}
            0 & W \\ W^* & 0
        \end{pmatrix}.
    \]
    Finally, the corresponding unitaries are given by $U_0^+ = - U_0^- = W^* | W |^{-1}$.
\end{proposition}

Note that the matrix $A$ is hermitian, hence is diagonalizable with real spectrum. As we will see in the proof, its spectrum is of the form $\sigma(A) = \{ \pm \mu_i\}_{1 \le i \le N}$ where $0 < \mu_1 \le \cdots \le \mu_N$ are the singular values of $W$. We used the notation $\1(A > 0)$ and $\1(A < 0)$ for the spectral projectors on $\R^+$ and $\R^-$ of $A$. 

The last part of this proposition states that $U_0^+$ is the inverse of the unitary appearing in the polar decomposition of $W$.

\begin{proof}
The first part is standard, and comes from the fact that $\Dirac$ is a Fourier multiplier by the hermitian matrix
\[
    D(k) := \begin{pmatrix}
    k & -\ri W \\ \ri W^* & -k
\end{pmatrix} \quad \in \cS_{2N}(\C), \qquad \text{so, in particular} \quad \sigma(\Dirac) = \bigcup_{k \in \R} \sigma\left( D(k) \right).
\]
For $E \in \R$, we have
\[
\det \left( D(k) - E \right)  = \det \begin{pmatrix}   k -E & - \ri W \\ \ri W^* & -k -E      \end{pmatrix} = \det \left( E^2 - k^2 - W^* W \right). 
\]
We deduce that $E \in \sigma(D(k))$ iff $E^2 - k^2 \in \sigma(W^* W)$ iff $E = \pm \sqrt{k^2 + \mu^2}$ for some $\mu^2 \in \sigma (W^* W)$, i.e. for some $\mu \ge 0$ a singular value of $W$. The result on the spectrum follows.

\medskip

We now focus on the energy $E = 0$. The matrix $D(k)$ is linear in $k$, hence analytic for $k \in \C$. If $\Psi \in \C^{2N}$ solves $D(k) \Psi = 0$ for some $k \in \C$, then, multiplying by $ \ri \sigma_3$, we get
\[
    \begin{pmatrix}
        \ri & 0 \\ 0 & -\ri 
    \end{pmatrix}
    \begin{pmatrix}
        k & - \ri W \\ \ri W^* & - k
    \end{pmatrix} \Psi = 0, \quad \text{that is} \quad
    \begin{pmatrix}
        0 & W \\ W^* & 0
    \end{pmatrix} \Psi = - \ri k \Psi,
    \quad \text{so} \quad (- \ri k) \in \sigma(A).
\]
Note that since $A$ is hermitian, we have $\sigma(A) \subset \R$, so the locations $k \in \C$ where $0 \in D(k)$ can only occur for $k$ on the imaginary axis. Let us study the matrix $A$. We write the polar decomposition of $W$ and $W^*$ as
\[
    W := | W | U, \quad \text{with} \quad | W | := \sqrt{W W^*}, 
\]
where $U \in \U(2N)$ is unitary. The polar decomposition of $W^*$ is related to the one of $W$, and we have
\[
    W^* := U^* | W| = | W^* | U^*, \quad \text{with} \quad | W^* | := \sqrt{W^*W} = U^* | W | U.
\]
Since $W$ is invertible, the decomposition is unique. Recall that $| W |$ and $|W^*| $ have the same spectrum, composed of the singular values of $W$. In what follows, we denote by $0 < \mu_1 \le \mu_2 \le \cdots \le \mu_N$ these singular values. Then, we can check that $A = \cU^* \Lambda \cU$, with
\[
     \cU = \frac{1}{\sqrt{2}}\begin{pmatrix}    1 & - U \\   -  U^* & 1
    \end{pmatrix}, \quad \text{and} \quad
    \Lambda = \begin{pmatrix}
        | W | & 0 \\ 0 & - | W^* |
    \end{pmatrix},
\]
and that $\cU \in \U(2N)$ is unitary. We deduce that $\sigma(A) = \sigma(\Lambda) = \{ \pm \mu_i \}_{1 \le i \le N}$. In addition, the spectral calculus shows that
\[
   \1(A > 0) =  \cU^* \begin{pmatrix}        1 & 0 \\ 0 & 0    \end{pmatrix} \cU 
    = \frac12 \begin{pmatrix}  1 & U \\  U^* & 1    \end{pmatrix}, \quad \text{so} \quad
    \ell_0^+ = \Ran \, \1(A > 0) = \left\{ \begin{pmatrix}
        x \\ U^* x
    \end{pmatrix}, \ x \in \C^N \right\}.
\]
By identification with~\eqref{eq:def:J0} (note that with our convention of Dirac, the symplectic form is given by the matrix $J$ in~\eqref{eq:def:J0}), we deduce as wanted that $U_0^+ = U^*$. The computation for $U_0^-$ is similar.


\end{proof}

\begin{remark}\label{rem:normal_W}
    In the case where $W$ is normal ($W W^* = W^* W$), we have $| W | = | W^* |$. In addition, $| W |$ and $U$ commutes.
\end{remark}

\subsection{Symmetries}

With this at hand, we can now identify some indices of $\Dirac$ solely from spectral property of the matrix $W$. We first record the following result, whose proof is straightforward (compare with~\eqref{eq:TCS_with_J}).
\begin{lemma} \label{lem:TCS_Symmetry_forV}
    A constant Dirac operator satisfies a $T$, $C$ and/or $S$ symmetry if, respectively,
    \[
        \forall k \in \R, \quad T D(k) = D(-k) T, \qquad C D(k) = -D(-k) C, \qquad S D(k) = -D(k) S.
    \]
    In terms of the matrix $J$ and $V := \begin{pmatrix}
        0 & -\ri W \\ \ri W^* & 0
    \end{pmatrix}$, this happens if $T J = J T$, $C J = -J C$ and $S J = -J S$, and
    \[
        T V = VT, \qquad CV = - VC, \qquad SV = - VS.
    \]
\end{lemma}
A $T$--symmetry implies $\sigma( D(k)) = D(-k)$, a $C$--symmetry implies $\sigma (D(k))  = - \sigma (D(-k))$ and an $S$--symmetry implies $\sigma (D(k)) = - \sigma (D(k))$. In the last case, the spectrum of $D(k)$ is symmetric with respect to $0$ for all $k \in \R$.

\medskip

We now identify the Index class by class. We focus on the classes for which the index is not vanishing, namely classes~$\rAIII$, $\rBDI$, $\rD$, $\rDIII$ and $\rCII$. Each of this class involves a $C$ or $S$ symmetry, so we focus on the energy $E = 0$. In what follows, we assume that $W$ is invertible, so the energy $E = 0$ is not in the spectrum of $\Dirac$.

\subsubsection{Classes $\rAIII$, $\rBDI$ and $\rCII$}
We first consider class $\rAIII$, which only involves a $S$--symmetry. As in Section~\ref{sec:class_AIII}, we choose a basis in which $S = \begin{pmatrix} 0 & 1 \\ 1 & 0 \end{pmatrix}$. The condition $SV = - VS$ reads, in terms of $W$, as
\[
    W = W^*.
\]
So $W \in \cS_N(\C)$. In particular, any such $W$ is normal (so $| W | = | W^* |$). Its polar decomposition $W = | W | U$ involves matrices $| W |$ and $U$ which commutes. With this notation, we have $U_0^+ = U^*$. The condition $W = W^*$ gives $U^* = U$, hence
\[
    \left( U_0^+ \right)^* := U_0^+,
\]
so $U_0^+ \in \U(N) \cap \cS_N(\C)$, as expected from class $\rAIII$. For this class, the index is the multiplicity of $1$ as an eigenvalue of $U_0^+$, hence the multiplicity of $-1$ as an eigenvalue of $U$, which is also the number of negative eigenvalues of $W$. We therefore proved the following Theorem.

\begin{theorem}[Protected edge modes in class $\rAIII$]
    Consider a Dirac operator of the form
    \[
    \Dirac = \begin{pmatrix}
        - \ri \partial_t & - \ri W(t) \\ \ri W(t) & \ri \partial_t
    \end{pmatrix} \quad \text{acting on} \quad L^2(\R, \C^{2N}),
    \]
    and with $W(t)$ a pointwise hermitian matrix. 
    \begin{itemize}
        \item If $W(t) = W_0$ is independent of $t \in \R$ and invertible, then the bulk--edge index of $\Dirac$ at $E = 0$ is the number of negative of $W_0$, noted $n^-(W_0)$.
        \item If $W(t)$ is a continuous junction between a constant potential $W_L$ on the left and $W_R$ on the right, with $W_L$ and $W_R$ invertible, then
        \[
        \dim \, \Ker \Dirac \ge \left| n^-(W_R) - n^-(W_L) \right| =  \left| n^+(W_R) - n^+(W_L) \right|.
        \]
    \end{itemize}
\end{theorem}
We recover the results in~\cite{Gom-}. In the case where $N = 1$, the function $W(t)$ can be seen as the varying mass for the Dirac operator. The theorem states that if the mass switches sign, then $0$ is in the kernel of the Dirac operator. This is a well-known result, which can be proved by directly solving the ODE, see for instance~\cite{FefLeeWei-17}. 

\medskip

We have similar results for classes $\rBDI$ and $\rCII$, and we do not repeat all the arguments here. For the class $\rBDI$, as in Section~\ref{sec:class_BDI}, we choose a basis in which 
\[
    S = \begin{pmatrix} 0 & 1 \\ 1 & 0 \end{pmatrix}, \quad 
    T = \begin{pmatrix} 0 & \cK \\ \cK & 0 \end{pmatrix}, \quad
    C =  \begin{pmatrix} \cK & 0 \\ 0 & \cK \end{pmatrix}.
\]
The $S$--symmetry is as before, and shows that $W$ must satisfy $W = W^*$. The $T$ or $C$--symmetry implies furthermore that $W = \overline{W}$ is real--valued. So the matrix $U$ (hence $U_0^+$) is also real--valued, and in particular $U_0^+ \in \O(N) \cap \cS_N(\R)$, as expected for the $\rBDI$ class.

\medskip

Finally, for the class $\rCII$, we must have $N = 2n$ even, and as in Section~\ref{sec:class_CII}, we choose a basis in which
\[
T =  - \begin{pmatrix}	0 &  \Omega \\ \Omega & 0 \end{pmatrix} \cK, \qquad 
C= \begin{pmatrix} \Omega & 0\\ 0 & \Omega \end{pmatrix}\cK, \qquad 
S = \begin{pmatrix}	0 & 1 \\ 1 & 0 \end{pmatrix}.
\]
This time, the $S$--symmetry still implies that the matrix $W$ satisfies $W = W^*$, and the $C$ or $T$--symmetry implies 
\[
    \Omega W = \overline{W} \Omega.
\]
We claim that this equality implies that $U_0^+$ is symplectic. First, we have
\[
    | \overline{W} |^2 = \overline{W} (\overline{W})^* = - \overline{W} \Omega \Omega (\overline{W})^* = -\Omega W W^*\Omega  = - \Omega | W |^2 \Omega,
\]
hence $| \overline{W} | = - \Omega | W | \Omega$ (since $-\Omega = \Omega^*$). On the other hand, it is classical that $\overline{| W |} = | \overline{W} |$. Altogether, we get 
\[
    \overline{| W |} = - \Omega | W | \Omega, \quad \text{or equivalently} \quad  \overline{| W |} \Omega = \Omega | W | .
\]
Finally, for the $U$ matrix (hence for the $U_0^+$ matrix), we get
\[
    \Omega \overline{U} = \Omega \overline{W} \overline{| W |^{-1}} = W | W |^{-1} \Omega = U \Omega,
\]
which is equivalent to $U \Omega U^T = \Omega$, that is $U \in \Sp(n)$ is symplectic.
\subsubsection{Class $\rD$}
In the $D$ class, we have a positive $C$--symmetry. As in Section~\ref{sec:class_D}, we can choose a basis in which $C = \cK$. The condition $CV = - VC$ reads $\overline{W}= W$, so $W$ is real--valued. So $U$ and $U_0^+$ are real-valued as well, and $U_0^+ \in \O(N)$, as expected. The Index reads
\[
    \Index^D(U_0^+) = \det(U_0^+) = {\rm sgn} \, \det(W).
\]

\subsubsection{Class $\rDIII$} In class $\rDIII$, we have a fermionic $T$ symmetry and a positive $C$ symmetry. As in Section~\ref{sec:class_DIII}, we work in a basis in which
$$
T =  \begin{pmatrix}
    0 & -\cK\\
    \cK & 0
\end{pmatrix}, \qquad C= \ri\cK, \qquad S = \begin{pmatrix}
    0 & \ri \\
    -\ri & 0
\end{pmatrix}.
$$
The conditions $TV = VT$ and $CV = - VC$ reads, in terms of $W$,
\[
    W = W^T, \quad W = \overline{W},
\]
so $W$ is real anti-symmetric, $W \in \cS^\R_N(\R)$. Recall that if such a matrix is invertible, then $N = 2n$ must be even. We deduce that $U_0^+ \in \O(N) \cap \cS_N(\R)$ is also real anti--symmetric, as expected from the $\rDIII$ class. In addition, 
\[
    \Index^\rDIII(U_0^+) = \Pf(U_0^+) =  {\rm sgn} \, \Pf(W^*) = (-1)^n {\rm sgn} \, \Pf(W).
\]

\appendix

\section{Appendix}

In this appendix, we gather some technical proofs.

\subsection{Proof of Leray's Theorem \ref{th:Leray}}
\label{sec:proofLeray}

We begin with the proof of Leray's theorem. Let us first check the vector subspace $\ell_U$ is Lagrangian. We set $V := \sqrt{ A_-}^{-1} U \sqrt{A_+}$. Note that $U$ is unitary iff $V$ satisfies $V^* A_- V =  A_+$. This gives
    \[
    \left\bra \begin{pmatrix}
        x \\ V x
    \end{pmatrix}, J \begin{pmatrix}
        y \\ V y
    \end{pmatrix} \right\ket_{\C^{2N}} = \ri \left[ \bra x, A_+ y \ket_{\C^N} - \bra Vx, A_- Vy \ket_{\C^N} \right]
    = \ri \bra x, (A_+ - V^* A_- V) u \ket_{\C^N} = 0.
    \]
    Conversely, let $\ell$ be a Lagrangian plane. Any $x \in \C^{2N}$ has a unique decomposition of the form $x = x_+ + x_-$ with $x_\pm \in K_\pm$. Consider the projection map $P_\ell : \ell \to K_+$ defined by $P_\ell(x) = x_+$. We claim that $P_\ell$ is one-to-one. First, since all Lagrangian planes are of dimension $N$ (this comes from the definition~\eqref{eq:def:Lagrangian_planes}), we have $\dim \, \ell = \dim \, K_+ = N$, so it is enough to check that $P_\ell$ is injective. Assume that there is $x_0 \in \ell$ so that $P_\ell(x_0) = 0$. Since $\ell$ is Lagrangian, for all $x, y \in \ell$, we have
    \begin{equation} \label{eq:identity_omega_A+_A-}
        0 = \omega(x, y) = \bra x, J y \ket_{\C^{2N}} = \ri \left[ \bra x_+, A_+ y_+ \ket_{\C^N} - \bra x_-, A_- y_- \ket_{\C^N} \right].
    \end{equation}
    In particular, for $x = y = x_0$, we get $\bra x_0, A_- x_0 \ket = 0$, hence $x_0 = 0$ as wanted, since $A_-$ is positive definite. Since $P_\ell$ is invertible, there is well-defined linear map $V_\ell$ from $K_+$ to $K_-$ so that, for all $x_+ \in K_+$, we have $x_+ + V_\ell(x_+) = P_\ell(x_+) \in \ell$. Together with~\eqref{eq:identity_omega_A+_A-}, we get that, for all $x_+, y_+ \in K_+ \approx \C^N$, we have
    \[
    \bra x_+, A_+ y_+ \ket_{\C^N} = \bra V_\ell(x_+), A_- V_\ell(y_+) \ket_{\C^N}.
    \]
    Since this holds for all $x_+, y_+ \in \C^N$, the matrix $V_\ell$ must satisfy the additional constraint that $A_+ = V_\ell^* A_- V_\ell$, which completes the proof.


\subsection{Proof of Lemma~\ref{lem:lEpm_are_Lagrangian_planes}}
\label{app:proof:lagrangian_planes}

We now provide a proof of Lemma~\ref{lem:lEpm_are_Lagrangian_planes}, which states that if $H$ is self--adjoint on $L^2(\R, \C^{2N})$ with $2N = Mp$ and $E \in \R \setminus \sigma(H)$, then $\ell_E^\pm$ are Lagrangian planes. The proof follows the line of~\cite[Theorem 27]{Gon-23}, but is simpler due to the fact that we work in finite dimension. Recall that we proved in~\eqref{eq:isotropic} that the sets $\ell_E^\pm$ are isotropic, and in particular, $\dim(\ell_E^\pm) \le N$.

\medskip

We now further assume that $E \notin \sigma(H)$, and claim that $\ell_E^+ \oplus \ell_E^- = \C^{2N}$. This will eventually prove that $\dim(\ell_E^+) + \dim(\ell_E^-) \ge 2N$, hence $\dim(\ell_E^\pm) = N$, and $\ell_E^\pm$ are Lagrangian.

Consider $x \in \C^{2N}$, and let $\Psi \in C^\infty_0(\R, \C^M)$ be any function so that $\Tr(\Psi) = x$. We set 
\[
    f := (H - E) \psi \ \in C^\infty_0(\R, \C^M), \qquad f^+ := \1_{\R^-} f, \quad f^- := \1_{\R^+} f.
\]
We used here that $H$ comes from a differential operator, which implies that $f$ is compactly supported (and in particular is in $L^2(\R)$). The functions $f^+$ and $f^-$ are square integrable as well, and satisfy $f = f^+ + f^-$. Since $E \notin \sigma(H)$, the operator $(H-E)^{-1}$ is bounded on $L^2(\R, \C^M)$, so the functions
\[
\Psi^+ := (H - E)^{-1}(f^+), \qquad \Psi^- := (H - E)^{-1} (f^-)
\]
are square integrable, and satisfy $\Psi = \Psi^+ + \Psi^-$. We set $x^+ := \Tr(\Psi^+)$ and $x^- := \Tr(\psi^-)$, so that $x = x^+ + x^-$. We claim that $x^+ \in \ell_E^+$. Indeed, consider $\phi^+$ the solution of the Cauchy problem $(\cL - E) \phi^+ = 0$ with initial data $x^+$ at $t = 0$. By definition, we have $\phi^+ \in \cS_E$. In addition, we remark that $(H - E) \Psi^+ = f_+$ which vanishes on $\R^+$. By uniqueness of the Cauchy problem, we deduce that $\phi^+$ coincides with $\Psi^+$ on $\R^+$, so $\phi^+$ is square--integrable at $+\infty$. This shows that $\phi^+ \in \cS_E^+$, hence that $x^+ \in \ell_E^+$. The proof is similar for $x^-$, and shows that $\ell_E^+ + \ell_E^- = \C^{2N}$. This concludes the proof of Lemma~\ref{lem:lEpm_are_Lagrangian_planes}

\subsection{The symplectic Grassmanian manifold \label{app:symp_grass}}
Finally, we give a short proof of the fact that
 $$
    \Sp(n) \cap \cS_{2n}(\C) \cong \bigcup_{m=0}^n \Sp(n) / \Sp(m) \times \Sp(n-m).
$$
This appears in class $\rCII$, and this case does not explicitly appear in~\cite{GonMonPer-22} (as was the case for the other classes). Consider $A \in \Sp(n) \cap \cS_{2n}(\C)$, that is
\[
    A = A^*  =A^{-1}, \quad \text{and} \quad A^T \Omega A = \Omega.
\]
Since $A$ is unitary hermitian, we have $\sigma(A) \in \{ \pm 1\}$. We denote by $M$ the mutiplicity of $1$. Let $P$ be the spectral projector on $E := \Ker(A - 1)$ (of dimension $M$). Note that $A = 2 P - 1$ is the orthogonal reflection with respect to this plane. The condition $A^T \Omega = \Omega A$ translates directly into
\begin{equation} \label{eq:symplectic_projector}
    P^T \Omega = \Omega P.
\end{equation}
This implies that $x \in E$ iff $\widetilde{x} := \Omega \overline{x} \in E$. In addition, we have $\langle \widetilde{x}, x \rangle = 0$, so the vectors $x$ and $\tilde{x}$ are always orthogonal. This shows that $M =\dim(E)$ is even, of the form $M = 2m$, and can be described by a set of $m$ orthogonal functions $(\phi_1, \cdots, \phi_m)$ so that $(\phi_1, \cdots, \phi_m, \phi_{m+1}, \cdots, \phi_{2m})$ is an orthonormal basis of $E$, where we set $\phi_{m+k} := \Omega \overline{\phi_k}$. We call such a basis a {\em symplectic frame} for $E$.

\medskip

On the other hand, a matrix $U$ is in $\Sp(n)$ iff it belongs to $\U(2n)$ and has of the block form
\begin{equation} \label{eq:form_symplectic_U}
    U = \begin{pmatrix}
        U_{11} & \overline{U_{12}}\\
        U_{12} & - \overline{U_{11}}
    \end{pmatrix}.
\end{equation}
Given such a matrix, with columns $U = (u_1, u_2, \cdots, u_{2n})$, and given $0 \le m \le n$, one can associate two symplectic frames of size $2m$ and $2(n-m)$ respectively, namely
\[
    \Phi = (u_1, u_2, \cdots, u_m ; u_{n+1}, \cdots u_{n+m}), \qquad 
    \Psi = (u_{m+1}, u_{m+2}, u_n ; u_{n+m+1}, \cdots u_{2n} ).
\]
So one can associate to any such $U \in \Sp(n)$ a matrix $A \in \Sp(n) \cap \cS_{2n}(\C)$ defined by $A := 2P - 1$, where $P$ denotes the orthogonal projector on ${\rm Ran}(\Phi)$. In addition, any matrix $A \in \Sp(n) \cap \cS_{2n}(\C)$ with $\dim \Ker(A - 1) = 1$ can be obtained by at least one $U \in \Sp(n)$ with this identification. Finally, two matrices $U_1$ and $U_2$ describes the same matrix $A$ iff the corresponding frames describe the same spaces: there is $U_\Phi \in \U(2 m)$ and $U_\Psi \in \U(2n-2m)$ so that
\[
    \Phi_1 = \Phi_2 U_\Phi \quad \text{and} \quad \Psi_1= \Psi_2 U_\Psi.
\]
This further implies that $U_\Phi := \Phi_2^* \Phi_1$ and $U_\Psi := \Psi_2^* \Psi_1$ are of the form~\eqref{eq:form_symplectic_U}, hence are symplectic matrices, so $U_\Phi \in \Sp(m)$ and $U_\Psi \in \Sp(n-m)$. The result follows.

\printbibliography

\end{document}